\documentclass[journal]{IEEEtran}
\usepackage{bbm}
\usepackage{caption}
\usepackage{subfigure}
\usepackage{graphicx}
\usepackage{latexsym}
\usepackage{diagbox}
\usepackage{changepage}
\usepackage[fleqn]{amsmath}
\usepackage{amsmath}
\usepackage{amsfonts}
\usepackage{indentfirst}
\usepackage{CJK}
\usepackage{indentfirst}
\usepackage[varg]{txfonts}
\usepackage{stfloats}
\usepackage{multirow}%
\usepackage{booktabs}
\usepackage{color,soul}
\usepackage{epstopdf}
\usepackage{float}
\usepackage{bm}
\usepackage{makecell}
\usepackage{array}
\usepackage{bm}
\usepackage{mathtools}
\usepackage{balance}
\usepackage{geometry}
\usepackage[linesnumbered,ruled,commentsnumbered,longend]{algorithm2e}
\usepackage[linesnumbered,ruled,commentsnumbered,longend]{algorithm2e}
\makeatletter

\newcommand{\Rmnum}[1]{\expandafter\@slowromancap\romannumeral #1@}

\makeatother
\makeatletter
\newtheorem{theorem}{Theorem}
\newtheorem{lemma}{Lemma}

\newenvironment{proof}[1][Proof]{\begin{trivlist}
		\item[\hskip \labelsep {\itshape #1}]}{\end{trivlist}}

\newcommand{\qed}{\nobreak \ifvmode \relax \else
	\ifdim\lastskip<1.5em \hskip-\lastskip
	\hskip1.5em plus0em minus0.5em \fi \nobreak
	\vrule height0.75em width0.5em depth0.25em\fi}

\SetAlgoLongEnd

\ifCLASSINFOpdf
\else
\fi

\begin{document}

\title{Secure Millimeter Wave Cloud Radio Access Networks Relying on Microwave Multicast Fronthaul}
\author{Wanming Hao,~\IEEEmembership{Member,~IEEE,} Gangcan Sun, Jiankang Zhang,~\IEEEmembership{Senior Member,~IEEE}, Pei Xiao,~\IEEEmembership{Senior Member,~IEEE,} and Lajos Hanzo,~\IEEEmembership{Fellow,~IEEE}
	\thanks{W. Hao is with the School of Information Engineering, Zhengzhou University, Zhengzhou 450001, China, and also with the Henan Institute of Advanced Technology, Zhengzhou University, Zhengzhou 450001, China,  and also with  the 5G Innovation Center, Institute of Communication Systems, University of Surrey, Guildford GU2 7XH, U.K.  (Email: iewmhao@zzu.edu.cn).}
	\thanks{G. Sun is with the School of Information Engineering, Zhengzhou University, Zhengzhou 450001, China  (E-mail: iegcsun@zzu.edu.cn).}
	\thanks{P. Xiao is with the Institute for Communication Systems (ICS), Home for 5G Innovation Centre, University of Surrey, Guildford, Surrey, GU2 7XH, U.K. (Email: p.xiao@surrey.ac.uk).}
	\thanks{J. Zhang and L. Hanzo are with the School of Electronics and Computer
		Science, University of Southampton, Southampton SO17 1BJ, U.K. (e-mail:
		jz09v@ecs.soton.ac.uk; lh@ecs.soton.ac.uk).}
}

%



\maketitle
\begin{abstract}
In this paper, we investigate the downlink secure beamforming (BF) design problem of cloud radio access networks (C-RANs) relying on multicast fronthaul, where millimeter-wave and microwave carriers are used for the access links and fronthaul links, respectively. The base stations (BSs) jointly serve users through cooperating hybrid analog/digital BF. We first develop an analog BF for cooperating BSs. On this basis, we formulate a secrecy rate maximization (SRM) problem subject both to a realistic limited fronthaul capacity and to the total BS transmit power constraint.  Due to the intractability of the non-convex problem formulated, advanced convex approximated techniques, constrained concave convex procedures and semi-definite programming (SDP) relaxation are applied to transform it into a convex one. Subsequently,  an iterative algorithm of jointly optimizing multicast BF, cooperative digital BF and the artificial noise (AN) covariance is proposed. Next, we construct the solution of the original problem by exploiting both the primal and the dual optimal solution of the SDP-relaxed problem. Furthermore, a per-BS transmit power constraint is considered, necessitating the reformulation of the SRM problem, which can be solved by an  efficient iterative algorithm. We then eliminate the idealized simplifying assumption of having perfect channel state information (CSI) for the eavesdropper links and invoke realistic imperfect CSI. Furthermore, a worst-case SRM problem is investigated. Finally, by  combining the so-called $\mathcal{S}$-Procedure and convex approximated techniques, we design an efficient  iterative algorithm to solve it. Simulation results are presented to evaluate the secrecy rate and demonstrate the effectiveness of the proposed algorithms. 
\end{abstract}


%
\IEEEpeerreviewmaketitle

\section{Introduction}
To satisfy the ever-increasing demand for higher data rates, cloud radio access networks (C-RANs) have been proposed as promising solutions, where the base stations (BSs) are connected to a central processor~(CP) such as the mobile switching center  through limited-capacity fronthaul links~\cite{1Peng},~\cite{2hao}. In C-RANs, the CP performs sophisticated baseband signal processing and resource optimization by exploiting global channel state information (CSI) for significantly reducing the interference among users, while improving the spectral efficiency (SE) of the system~\cite{3Yu}. Additionally, with the deployment of ultra-dense BSs, it becomes beneficial  for the adjacent BSs to form a BS cluster and to cooperatively/jointly serve users relying on the sophisticated coordinated multiple-point (CoMP) transmission approach~\cite{4CoMP}.

The CoMP-based C-RAN structure of Fig.~\ref{systemfigure}  generally includes a CP, cooperating BSs and users. The BSs receive data from the CP via the fronthual links, while the users are jointly served by the BSs via cooperative beamforming (BF). Thus, a pair of fundamental problems have to be considered, namely: the fronthaul/access links and the data sharing approach. For the former, in general, there are multiple carrier technologies that can be adopted, such as tethered links relying on fiber or copper, and wireless links using microwave or millimeter-wave (mmWave) carriers~\cite{4Li}-\cite{6Zhang}.  However, for the ultra-dense BS deployment, it would be unrealistic to have a wired connection between each BS and the CP due to its high cost. Furthermore, the capacity provided by wired links is fixed, hence it cannot be dynamically adjusted to cater for  the traffic variation according to the users' demands. Therefore, the more flexible wireless carriers constitute appropriate candidates for fronthaul links~\cite{5Wireless}. Here, we adopt the microwave fronthaul-link carriers and mmWave access-link carriers, and the reason as follows: The distance between the CP and the BSs is relatively high, hence it is inefficient to use mmWave frequencies due to their high propagation loss. Therefore we adopted a microwave fronthaul as a benefit of its lower path-loss. Additionally, although the propagation loss of mmWave signals is high, the distance between the BSs and users
is relatively short, hence it is appropriate to use mmWave access links. Furthermore, these
choices avoid the interference between the fronthaul and access links.  As for the data sharing approach among the CP and BS cluster, when multiple BSs jointly serve a user, the CP needs to transmit this user's message to all the cooperating BSs. Fortunately, this point-to-multipoint CP-BSs fronthaul transmission can be readily realized by multicast BF, which has been applied in~\cite{6Hu},~\cite{7Yu}. Therefore, in this paper, we will formulate a more practical system relying on  BS-cooperation-aided  mmWave C-RAN combined with a microwave multicast fronthaul.

Additionally, due to the broadcast nature of wireless communication, the users' confidential messages may be eavesdropped, jeopardizing their secure transmissions~\cite{8Wang}-\cite{10Yang}. As a compelling solution, physical layer security (PLS)  has been proposed for enhancing the  communication security~\cite{11Chu}.  Although numerous solutions have been conceived for  PLS~\cite{8Wang}-\cite{11Chu}, the following challenges exist in  our proposed system structure. {\textit{i)}} Hybrid analog/digital BF design: The BSs are usually equipped with multiple mmWave antennas to compensate for the serious propagation loss. However, in order to reduce the radio frequency (RF) energy consumption and hardware cost, only a few RF chains are affordable~\cite{12Dai}. For this reason, we consider  hybrid analog/digital BF designs for the cooperating BSs.  {\textit{ii)}} Artificial noise (AN) design: To further improve the transmission security, the AN used for jamming should be properly designed for the cooperating BSs.  {\textit{iii)}} Multicast BF design: The CP-BSs multicast BF has to be optimally designed for improving the fronthaul capacity. 

Against this background, in this paper, we investigate the downlink (DL) transmission security problem of a mmWave C-RAN with multicast fronthaul. Our main contributions include: 
\begin{itemize}
\item  We propose a BS-cooperation-aided C-RAN relying on a microwave multicast fronthaul. The adjacent BSs form a cooperating cluster and jointly serve the users relying on mmWave carriers, while the BSs receive their fronthaul data from the CP with microwave multicast. To reduce both the hardware cost and energy consumption, each BS is equipped with  a single RF chain connected to  multiple antennas. This scenario is further complicated by the fact that the eavesdroppers (Eves) maliciously attack the legitimate users in order to wiretap their confidential messages. 
\item  We design a DL analog transmit BF scheme for the cooperating BSs. Based on an equivalent channel model,  we formulate a DL secrecy rate maximization (SRM) problem  by jointly optimizing the microwave CP-BSs multicast BF, the cooperative mmWave BSs-users digital BF and the AN covariance under the constraints of fronthaul capacity and  total BS transmit power. However, the  problem formulated is intractable because of  its non-convexity. Hence,  advanced convex approximated techniques, constrained concave convex procedures (CCCP) and semi-definite programming (SDP) relaxation are applied to recast the original problem into a convex one. Then, we propose an efficient iterative algorithm for solving the original problem.  Meanwhile, we construct the solution of the original problem by exploiting both the primal and the dual optimal solution of the SDP-relaxed problem. Next, we replace the total BS transmit power constraint by a  per-BS transmit power constraint and reformulate the SRM problem. An iterative solution algorithm is also proposed.

\item To consider a practical scenario, we then dispense with the idealized simplifying assumption of having perfect CSIs for the Eves links and replace it by realistic imperfect CSIs ones followed by investigating a worst-case SRM problem. The problem formulated is then solved by our proposed iterative algorithm relying on  the classic $\mathcal{S}$-Procedure and the convex approximated techniques. Our simulation results demonstrate the efficiency of the proposed algorithms.
\end{itemize}

The rest of this paper is organized as follows. The related contributions are summarized in~Section II. The systems description and the analog BF designed are presented in Section~III. In Section IV, the SRM problem formulated both under total BS and under per-BS  transmit power constraints are solved, respectively. A realistic imperfect CSI is considered  for the Eves links in Section V. Our simulation results are drawn in Section~VI. Finally, we conclude this paper in Section~VII.

\textit{Notations}: We use the following notations throughout this paper: 
$(\cdot)^T$ and $(\cdot)^H$ denote the transpose and Hermitian transpose, respectively, $\|\cdot\|$ is the Frobenius norm, ${\mathbb{C}}^{x\times y}$ means the space of $x\times y$ complex matrix, {Re($\cdot$)} and {Tr($\cdot$)} denote real number operation and trace operation, respectively. $[\cdot]^+$ denotes the $\max\{0,\cdot\}$, and Diag($a,\ldots,a$) is a diagonal matrix.
\section{Related Works}
\subsection{Physical Layer Security in mmWave Communications}
In~\cite{13Wang}, Wang~{\it{et al.}} invoke both the maximum ratio transmit BF and AN BF  for improving the secrecy capacity, and derive the closed-form expression of the connection probability. Additionally, the throughput is maximized under outage probability constraints. Alotaibi and Hamdi~\cite{14Hamdi} propose a switched phased-array antenna structure for enhancing the secrecy transmission. On this basis,  the authors analyze the secrecy capacity and derive the exact expression of the bit error probability. To reduce the hardware cost, Huang~{\it{et al.}}~\cite{15Huang} consider a pair of  sparse RF chain structures and design a directional hybrid analog/digital precoding for substantially  improving the secrecy transmission capacity. Relying on a uniform linear array at the BS, Zhu~{\it{et al.}} propose a tractable technique for evaluating the secrecy rate, when exploiting the AN~\cite{16Heath}. Their research results show that the secrecy performance can be improved at low transmit powers. Sun~{\it{et al.}}~\cite{17Tang}  investigate the secrecy capacity in a simultaneous wireless information and power transfer (SWIPT) network associated with unmanned aerial vehicular (UAV) relays. The closed-form expression of the average secrecy transmission rate is derived for a  three-dimensional mmWave antenna model and then,  the lower bound of the secrecy transmission rate is maximized. Similarly, Sun~{\it{et al.}}~\cite{18Tao} consider a secrecy SWIPT-aided mmWave ultra-dense network and derive the energy-information coverage probability as well as the effective secrecy throughput both under power splitting and time switching with the aid of stochastic geometry. Wang~{\it{et al.}}~\cite{19Jing} consider the PLS of relay-aided mmWave systems and design a two-stage secure hybrid analog/digital precoding scheme. 

Although the PLS of mmWave communications has been lavishly investigated in~\cite{13Wang}-\cite{19Jing}, most authors only analyze the secrecy capacity and derive the related  closed-form expressions, such as~\cite{13Wang},~\cite{14Hamdi},~\cite{17Tang}, and~\cite{18Tao}. However, there is a paucity of information on the associated BF design and resource optimization. The authors of \cite{15Huang} and~\cite{19Jing} do consider the hybrid analog/digital precoding design problem under the idealized assumption of perfect CSI. 
\subsection{Multicast Transmissions in Cloud Radio Access Network}
Hu~{\it{et al.}}~\cite{6Hu} propose to transmit the users' data from the CP relying on fronthaul multicast, and formulate a joint BF design and  user clustering problem for maximizing the weighted sum rate. Then, an iterative binary search algorithm is proposed.  A cache-based C-RAN using a multicast backhaul/fronthaul is investigated by Dai~{\it{et al.}}~\cite{7Yu}. They also formulate a joint BS cache  and BF design problem for minimizing the transmission delay and then propose a powerful optimization algorithm.  Chen~{\it{et al.}}~\cite{20He} investigate a BF design problem in a multigroup multicast C-RAN relying on limited backhaul capacity, and a robust BF design algorithm is developed for maximizing the worst-user rate.  Similar to~\cite{20He}, Tao~{\it{et al.}}~\cite{21Tao} assume that the users will form a cluster when they request the same content from the CP, hence they are served by the cache-based BS cluster through multicast. They also propose a dynamic BS cache and multicast  BF design scheme for minimizing the weighted backhaul cost and transmit power.  As a future advance, Vu~{\it{et al.}}~\cite{23Vu}  conceive joint BS selection and BF design for maximizing the minimum weighed data rate under finite-capacity fronthaul links. Based on this,  they also develop an upper bound  based on the classic semidefinite relaxation technique and propose a heuristic low-complexity iterative algorithm. Shi~{\it{et al.}}~\cite{24Zhang} investigate the power minimization problem by dynamically selecting active BSs under imperfect CSI, and a  robust and sparse three-stage BF design algorithm is proposed. To minimize the delivery latency in a cache-based multigroup multicasting network, He~{\it{et al.}}~\cite{25He} develop three transmission schemes and evaluate their performance. 

Indeed, substantial research efforts have been invested in designing multicast-aided C-RANs. However, most of them consider  fixed-capacity/wired fronthaul links, such as~\cite{20He}-\cite{25He}, and multicast transmission is only applied to the access links. Although the authors of~\cite{6Hu} and~\cite{7Yu} do investigate the multicast fronthaul/backhaul-aided C-RAN, they mainly focus their attention on the caching and transmission delay problem. Furthermore,  the specifics of mmWave communication and PLS are  not considered in the above contributions.

In contrast to the aforementioned studies, we investigate the transmission security problem in a BS-cooperation-aided mmWave C-RAN relying on a limited-capacity microwave multicast fronthaul.  {In the system considered, there are two techniques conceived for  securing the communication link: $i)$ Carefully constructed antenna beamforming, and $ii)$ Jamming by artificial noise. More explicitly, beamforming improves the SINR by mitigating the interference, while artificial noise is used for confusing the eavesdropper so as to improve the secure rate.}  {The main contributions of this paper can be summarized from the following three perspectives: $i)$ Propose a BS-cooperation-aided C-RAN relying on a microwave multicast fronthaul; $ii)$ Develop a joint downlink transmit beamforming design algorithm (i.e., Algorithm 2) for the cooperating BSs under perfect CSI, and analyze the rank characteristic of the solution as well as give the corresponding proof; $iii)$ Investigate a worst-case sum rate maximization problem by considering imperfect CSI knowledge, and propose an efficient beamforming design algorithm (i.e., Algorithm 3).}
\section{System Model and Analog Beamforming Design}

In this section, we first describe the system model and mmWave channel model. Then, we propose an analog BF design scheme for cooperating BSs. 
\subsection{System Model}
A DL C-RAN relying on a central CP is considered, as shown in Fig.~\ref{systemfigure}. For simplicity,  we consider a single BS cluster consisting of $L$ cooperating BSs\footnote{Although we consider a single BS cluster, the proposed algorithm and analytical results can be extended to multiple BS clusters.}, where all BSs jointly serve $K$ single-antenna users via cooperative mmWave BF. We assume that there are $Z$ single-antenna Eves potentially overhearing the users' message.  The CP equipped with $N$ antennas transmits its signals to the BSs via microwave multicast fronthaul  links, and each BS relies on a single receive antenna (RA) and  $M$ transmit antennas (TAs). Since the BS's TAs operate at a mmWave frequency, we assume that only a single RF chain serves $M$ TAs through a group of quantized phase shifters for reducing both the hardware cost and  energy consumption\footnote{Since each BS is equipped with a single RF chain, to reduce the inter-user interference, we assume that the number of users being served simultaneously is small than that of BSs. 
	When there are are users, they can be served by time division multiplex access
	(TDMA) or by appropriate user scheduling~\cite{add1}.}~\cite{26Hao}. {Additionally,  the wireless fronthaul links rely on out-of-band spectrum, hence there is no interference between the access and fronthaul links.}
\begin{figure}[t]
	\begin{center}
		\includegraphics[width=8cm,height=6cm]{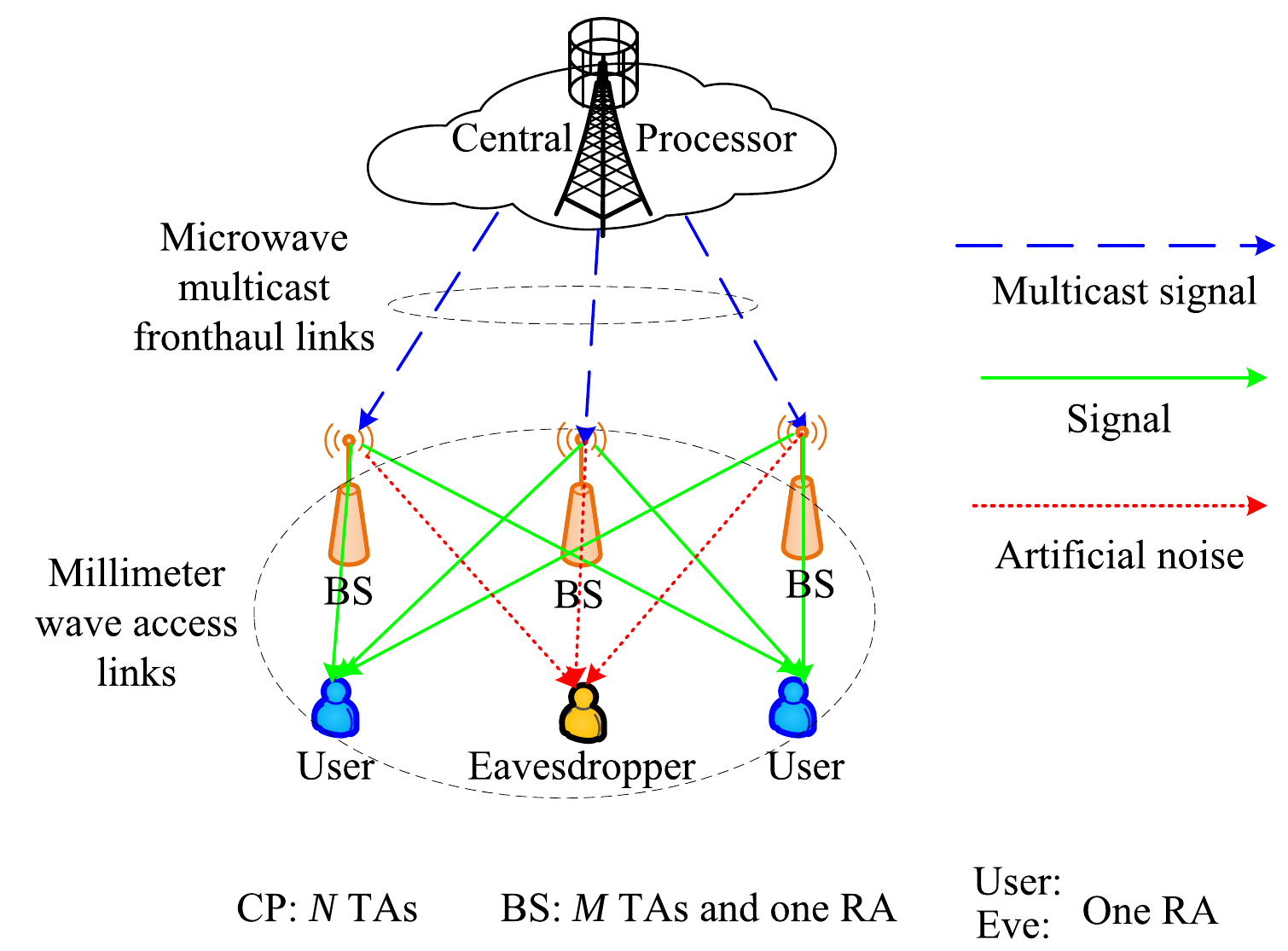}
		\caption{System model of the BS-cooperation-aided C-RAN DL using multicast fronthaul.}
		\label{systemfigure}
	\end{center}
\end{figure}
\subsubsection{Microwave Multicast Fronthaul Link} The  signal received at the $l$th BS can be written as
\begin{eqnarray}
y^{\rm{FH}}_l={\bf{g}}_l{\bf{v}}_0x_0+n_l,
\end{eqnarray}
where ${\bf{g}}_l\in {\mathbb{C}}^{1\times N}$ denotes the fronthaul link's channel vector from the CP to the $l$th BS,  ${\bf{v}}_0\in {\mathbb{C}}^{N\times 1}$ represents the multicast BF vector  transmitted by the CP to the BS cluster, and $x_0$ is the multicast signal with ${\mathbb{E}}\{|x_0|^2\}=1$. Furthermore, $n_l$ is the independent and identically distributed (i.i.d.) additive white Gaussian noise (AWGN)  with zero mean value and follows $\mathcal{CN}(0,N_0)$. 

As a result, the achievable fronthaul rate of the $l$th BS is
\begin{eqnarray}
R_l^{\rm{FH}}=W_{\rm{mc}}\log\left(1+\frac{|{\bf{g}}_l{\bf{v}}_0|^2}{W_{\rm{mc}}N_0}\right),
\end{eqnarray} 
where $W_{\rm{mc}}$ denotes the DL microwave bandwidth. Since the fronthaul multicast rate is limited by the BS owning the worst channel condition, the fronthaul rate provided by the CP can be expressed as
\begin{eqnarray}
R_{\rm{FH}}=\underset{l\in{\mathcal{L}}}{\min}\;\;\;\left\{R_l^{\rm{FH}}\right\},
\end{eqnarray} 
where $\mathcal{L}=\{1,\dots,L\}$ denotes the BS set. 
\subsubsection{Millimeter Wave Access Link} The received signal of the $k$th user can be written as 
\begin{eqnarray}
\begin{aligned}
 y_k^{\rm{AC}}=&{\bf{h}}_k{\bf{F}}\left(\sum\nolimits_{i=1}^K{\bf{v}}_ix_i+{\bf{q}}\right)+n_k\\
 =&\underbrace{{\bf{h}}_k{\bf{F}}{\bf{v}}_kx_k}_{\rm{Desired\;signal}}+\underbrace{{\bf{h}}_k{\bf{F}}\sum\nolimits_{i\neq k}^K{\bf{v}}_ix_i}_{\rm{Interference}}
 +\underbrace{{\bf{h}}_k{\bf{F}}{\bf{q}}}_{\rm{AN}}+\underbrace{n_k}_{\rm{Noise}},
 \end{aligned}
\end{eqnarray}
where ${\bf{h}}_k=[{\bf{h}}_k^1,\dots,{\bf{h}}_k^L]$ represents the DL channel vector from $L$ cooperating BSs to the $k$th user,  ${\bf{h}}_k^l\in{\mathbb{C}}^{1\times M}$ represents the DL channel vector from the $l$th BS to the $k$th user,  while ${\bf{q}}\in{\mathbb{C}}^{L\times 1}$ is the AN transmitted by the BS cluster. It is assumed that ${\bf{q}}\in \mathcal{CN}(0,{\bf{\Lambda}})$, where ${\bf{\Lambda}}$ denotes the AN covariance matrix to be optimized, while ${\bf{v}}_k\in{\mathbb{C}}^{L\times 1}$ and $x_k$, respectively, denote the digital BF vector and the signal desired for the $k$th user, and  $n_k$ is the i.i.d. AWGN  with $\mathcal{CN}(0,N_0)$.  ${\bf{F}}\in{\mathbb{C}}^{ML\times L}$ is the analog BF, which can be expressed~as
\begin{eqnarray}\label{AG}
	{\bf{F}}=\left[
	\begin{array}{cccc}
	{\bf{f}}_{1} & {\bf{0}}&\cdots & {\bf{0}} \\
	{\bf{0}}&{\bf{f}}_{2} & \cdots & {\bf{0}}\\
	\vdots & \vdots & \ddots&\vdots \\
	{\bf{0}} & {\bf{0}} &\cdots & {\bf{f}}_{L} \\
	\end{array}
	\right],
	\end{eqnarray}
where ${\bf{f}}_{l}\in{\mathbb{C}}^{M\times 1}$ denotes the analog BF vector designed by the $l$th BS and all elements of ${\bf{f}}_{l}$ have the same amplitude but different phases~\cite{27Dai}, namely $|{\bf{f}}_{l}(m)|=1/\sqrt{M}\;(m\in{\mathcal{M}})$, where ${\mathcal{M}}=\{1,\dots,M\}$ is the antenna set of each BS and ${\bf{f}}_{l}(m)$ represents the $m$th element of ${\bf{f}}_{l}$.

Accordingly, the achievable rate of the $k$th user can be expressed as
\begin{eqnarray}\label{ACRate}
	R_k^{\rm{AC}}\!=\!W_{\rm{mm}}\log\left(1\!+\!\frac{|{\bf{h}}_k{\bf{F}}{\bf{v}}_k|^2}{\sum\nolimits_{i\neq k}^K|{\bf{h}}_k{\bf{F}}{\bf{v}}_i|^2\!+\!{\bf{h}}_k{\bf{F}}{\bf{\Lambda}}({\bf{h}}_k{\bf{F}})^H\!+\!W_{\rm{mm}}N_0}\right),
\end{eqnarray} 
where $W_{\rm{mm}}$ denotes the mmWave bandwidth.

On the other hand, the Eves try to intercept the message of the $k$th user, and the  signal received at the $z$th Eve can be written as
\begin{eqnarray}
\begin{aligned}
y_z^{\rm{EV}}=&{\bf{h}}_z^{\rm{e}}{\bf{F}}\left(\sum\nolimits_{i=1}^K{\bf{v}}_ix_i+{\bf{q}}\right)+n_z\\
=&\underbrace{{\bf{h}}_z^{\rm{e}}{\bf{F}}{\bf{v}}_kx_k}_{\rm{Desired\;signal}}+\underbrace{{\bf{h}}_z^{\rm{e}}{\bf{F}}\sum\nolimits_{i\neq k}^K{\bf{v}}_ix_i}_{\rm{Interference}}+\underbrace{{\bf{h}}_z^{\rm{e}}{\bf{F}}{\bf{q}}}_{\rm{AN}}+\underbrace{n_z}_{\rm{Noise}},
\end{aligned}
\end{eqnarray}
where ${\bf{h}}_z^{\rm{e}}=[{\bf{h}}_z^{\rm{e}1},\dots,{\bf{h}}_k^{{\rm{e}}L}]$ denotes the DL channel vector from $L$ cooperating BSs to the $z$th Eve,  ${\bf{h}}_z^{{\rm{e}}l}\in{\mathbb{C}}^{1\times M}$ represents the DL channel vector from the $l$th BS to the $z$th Eve. Then, the $z$th Eve can wiretap the $k$th user at a rate of 
 \begin{eqnarray}
 R_{k,z}^{\rm{EV}}\!=\!W_{\rm{mm}}\log\left(1\!+\!\frac{|{\bf{h}}_z^{\rm{e}}{\bf{F}}{\bf{v}}_k|^2}{\sum\nolimits_{i\neq k}^K|{\bf{h}}_z^{\rm{e}}{\bf{F}}{\bf{v}}_i|^2\!+\!{\bf{h}}_z^{\rm{e}}{\bf{F}}{\bf{\Lambda}}({\bf{h}}_z^{\rm{e}}{\bf{F}})^H\!+\!W_{\rm{mm}}N_0}\right).
 \end{eqnarray} 
 
Finally, the achievable secrecy rate of the $k$th user can be expressed as~\cite{28Zhou}
\begin{eqnarray}
R_{k}^{\rm{SR}}=\left[R_k^{\rm{AC}}-\underset{z\in{\mathcal{Z}}}{\max}\left\{R_{k,z}^{\rm{EV}}\right\}\right]^+,
\end{eqnarray}
where ${\mathcal{Z}}=\{1,\dots,Z\}$ denotes the Eves set.
\subsubsection{Millimeter Wave Channel Model}
In this paper, we adopt  the mmWave channel model having $C$ scatters, where each scatter includes a single propagation path between the $l$th BS and the $k$th user~\cite{29Heath}. Therefore, the mmWave channel ${\bf{h}}_{k}^l$ can be described~as
\begin{eqnarray}
{\bf{h}}_{k}^l=\sqrt{\frac{M}{C}}\sum_{c=1}^C\alpha_{k,l}^c{\bf{a}}(\theta_{k,l}^c),
\end{eqnarray}
where $\alpha_{k,l}^c$ denotes the complex gain of the $c$-th path, $\theta_{k,l}^c\in[0,\pi]$ is the azimuth angle of arrival for the $c$-th path, and ${\bf{a}}(\theta_{k,l}^c)$ represents the antenna array's steering vector that can  be expressed~as
\begin{eqnarray}\label{eq5}
{\bf{a}}(\theta_{k,l}^c)={\frac{1}{\sqrt{M}}}\left[1,e^{j\frac{2\pi}{\lambda}d\sin(\theta_{k,l}^c)},\ldots, e^{j\frac{2\pi}{\lambda}(M-1)d\sin(\theta_{k,l}^c)}\right],
\end{eqnarray}
where $d$ and $\lambda$, respectively, denote the inter-antenna distance and signal wavelength. The channel model of Eves has a similar expression  to (10), which is omitted it here.
\subsection{Analog Beamforming Design For Cooperating BSs}
In order to reduce the total complexity of the joint BF (including the analog, digital and multicast BF) design,  we first design the analog BF (i.e., $\bf{F}$), and then obtain the equivalent mmWave channel. In  practice, only a quantized phase can be realized~\cite{30Dai}. We assume that $B$-bit quantized phase shifters are used, and the non-zero elements of the analog BF $\bf{F}$ should belong to
\begin{eqnarray}\label{phase}
	\frac{1}{\sqrt{M}}\left\{e^{j\frac{2\pi \phi}{2^B}}:\phi=0,1,\dots,2^{B-1}\right\}.
\end{eqnarray}

According to~(\ref{AG}), we need to design the analog BF for $L$ BSs, respectively. For the $k$th user, we have ${\bf{h}}_k{\bf{F}}=[{\bf{h}}_k^l{\bf{f}}_1,\dots,{\bf{h}}_k^L{\bf{f}}_L]$, and thus we can maximize the array again $|{\bf{h}}_k^l{\bf{f}}_l|$ by appropriately selecting the optimal quantized phase from~(\ref{phase}). For example, for the analog BF vector ${\bf{f}}_l$  of the $l$th BS, the $m$th angle of ${\bf{f}}_l$ can be selected as
\begin{eqnarray}\label{analog}
		\hat{\phi}=\underset{\phi\in\{0,1,\dots,2^{B-1}\}}{\arg {\rm{min}}}\left|{{\angle}}({\bf{h}}_{k}^l(m))-\frac{2\pi \phi}{2^B}\right|,
\end{eqnarray}
where ${{\angle}}(\cdot)$ denotes the angle of $(\cdot)$, and we have 
\begin{eqnarray}
	{\bf{f}}_l(m)=\frac{1}{\sqrt{M}}e^{j\frac{2\pi \hat{\phi}}{2^B}}.
\end{eqnarray}

Additionally, to guarantee user fairness, the BF designed  for the $L$ BSs should not only  maximize the single user's array gain, and each user should be allocated at least one BS to maximize its array gain. We summarize the analog BF design scheme in Algorithm~\ref{algorithm1}. 

\begin{algorithm}[t]
	{{\caption{Analog Beamforming Design for $L$ cooperating BSs.}
			\label{algorithm1}
			{\bf{Initialize}} ${\bf{F}}={\bf{0}}_{{ML}\times L}$, $l=1$.\\
			\While {$l\leq L$}
			{\For{$k=1:K$}{
					\For{$m=1:M$}{Solve (\ref{analog}) and obtain ${\bf{f}}_l(m)=\frac{1}{\sqrt{M}}e^{j\frac{2\pi \hat{\phi}}{2^B}}.$}
					${\bf{F}}((l\!-\!1)M\!+\!1:Ml,l)={\bf{f}}_l$, $l\leftarrow l+1$.\\
					{\bf{{If}}} {$l>L$} {\bf{{then}}} Break.}}}
	}
\end{algorithm}
\section{Secrecy Rate Maximization Problem Formulation and Solution}
In this section, we first formulate the SRM problem under the total BS transmit  power as well as  CP transmit power constraints, and an efficient iterative algorithm is proposed. Then, a per-BS transmit power constraint is considered and a similar SRM algorithm is developed.
\subsection{Problem Formulation Under Total BS Transmit Power Constraint}
After designing the analog BF ${\bf{F}}$, we can achieve the equivalent channel of the $k$th user as $\overline{{\bf{h}}}_k={\bf{h}}_k{\bf{F}}$. Similarly, the equivalent channel of the $z$th Eve can be expressed as $\overline{{\bf{h}}^e}_z={\bf{h}}^e_k{\bf{F}}$.

We assume that the multicast fronthaul transmission time frame includes $K$ slots, and each slot is used to transmit a single user's message from the CP to $L$ cooperating BSs. For simplicity, the frame length and the $k$th slot length are, respectively, assumed to be 1 and $t_k$, thus we have $\sum_{k=1}^{K}t_k=1$. According to the specific constraint that the achievable capacity of the $k$th user must be lower than the fronthaul capacity provided by the CP for the $k$th user~\cite{32hao}, we have   
\begin{eqnarray}\label{fhcapacity}
R_k^{\rm{AC}}\leq t_kR_{\rm{FH}}, k\in\mathcal{K},
\end{eqnarray}
where $\mathcal{K}=\{1,\cdots,K\}$ denotes the user set. {Since $\sum_{k=1}^{K}t_k=1$ and $R_{\rm{FH}}=\underset{l\in{\mathcal{L}}}{\min}\;\;\;\left\{R_l^{\rm{FH}}\right\}$, we can obtain (\ref{fhcapacity}) as follows:
\begin{eqnarray}\label{fhconstraint}
\sum_{k=1}^{K}R_k^{\rm{AC}}\leq \sum_{k=1}^{K}t_kR_{\rm{FH}}=\underset{l\in{\mathcal{L}}}{\min} \left\{R_l^{\rm{FH}}\right\}\sum_{k=1}^{K}t_k=\underset{l\in{\mathcal{L}}}{\min} \left\{R_l^{\rm{FH}}\right\}.
\end{eqnarray}}
Next, we have to state that as long as there is at least one scenario in which we can obtain (\ref{fhcapacity}) from (\ref{fhconstraint}), then we can find one case $R_k^{\rm{AC}}\leq t_kR_{\rm{FH}}, k\in\mathcal{K}$ from (\ref{fhconstraint}) to satisfy the fronthaul capacity constraint. Next, we provide the proof. We rewrite (\ref{fhconstraint}) as follows:
\begin{eqnarray}\label{fhconstraintA}
\sum_{k=1}^{K}R_k^{\rm{AC}}\leq \underset{l\in{\mathcal{L}}}{\min} \left\{R_l^{\rm{FH}}\right\}.
\end{eqnarray}
Since $\sum_{k=1}^{K}t_k=1$ and $R_{\rm{FH}}=\underset{l\in{\mathcal{L}}}{\min}\;\;\;\left\{R_l^{\rm{FH}}\right\}$, we have 
\begin{eqnarray}\label{fhconstraintB}
\begin{aligned}
\underset{l\in{\mathcal{L}}}{\min} \left\{R_l^{\rm{FH}}\right\}=&\underset{l\in{\mathcal{L}}}{\min} \left\{R_l^{\rm{FH}}\right\}\cdot 1=\underset{l\in{\mathcal{L}}}{\min} \left\{R_l^{\rm{FH}}\right\}\sum_{k=1}^{K}t_k\\
=&\sum_{k=1}^{K}t_kR_{\rm{FH}}\geq \sum_{k=1}^{K}R_k^{\rm{AC}}.
\end{aligned}
\end{eqnarray} 
According to (\ref{fhconstraintB}), we have 
\begin{eqnarray}\label{fhconstraintC}
t_1+t_2+...+t_K\geq \frac{R_1^{\rm{AC}}}{R_{\rm{FH}}}+\frac{R_2^{\rm{AC}}}{R_{\rm{FH}}}+...+\frac{R_K^{\rm{AC}}}{R_{\rm{FH}}}.
\end{eqnarray} 
Hence, there always exists one case satisfying
\begin{eqnarray}\label{fhcapacityA}
R_k^{\rm{AC}}\leq t_kR_{\rm{FH}}, k\in\mathcal{K}.
\end{eqnarray} 
which completes the proof.

Finally, we formulate the joint BF and AN variance design problem to maximize the secrecy rate, which can be cast as
\begin{subequations}\label{OptA}
	\begin{align}
&\underset{\left\{{\bf{v}}_0,\{{\bf{v}}_k\},{\bf{\Lambda}}\right\}}{\rm{max}}\;\sum_{k=1}^{K}\left[R_k^{\rm{AC}}-\underset{z\in{\mathcal{Z}}}{\max}\left\{R_{k,z}^{\rm{EV}}\right\}\right]^+ \label{OptA0}\\
	{\rm{s.t.}}&\;\; \sum_{k=1}^{K}R_k^{\rm{AC}}\leq \underset{l\in{\mathcal{L}}}{\min}\;\left\{R_l^{\rm{FH}}\right\},\label{OptA1}\\
	           &\;\;\sum_{k=1}^{K}||{\bf{F}}{\bf{v}}_k||^2+{\rm{Tr}}({\bf{F}}^H{\bf{F}}{\bf{\Lambda}})\leq P_{\rm{max}}^{\rm{BS}},\label{OptA2}\\
	           &\;\;||{\bf{v}}_0||^2\leq P_{\rm{max}}^{\rm{AC}},\label{OptA3}
	\end{align}
\end{subequations} 
where (\ref{OptA1}) denotes the fronthaul capacity constraint, (\ref{OptA2}) is the total  transmit power constraint of the $L$ cooperating BSs, and (\ref{OptA3}) represents the CP transmit power constraint. However, it is a challenge to directly solve (\ref{OptA}) due to the non-convex nature of the objective function (OF)~(\ref{OptA0}) and constraint (\ref{OptA1}).
\subsection{Problem Solution}  
{It is clear that the secrecy rate should not be lower than zero during optimization, otherwise the solution is not meaningful. Therefore, we can move the operation $[\;]^+$ of the OF~(\ref{OptA0})}\footnote{{Although  $[\;]^+$ is still concave function~\cite{add2}, we move it for convenience.}}. By introducing the auxiliary variables $\{\beta_k\}$ and $\{\gamma_k^z\}$, (\ref{OptA0}) can be transformed as follows:
\begin{eqnarray}\label{Obj1}
\sum_{k=1}^{K}W_{\rm{mm}}\left(\log(1+\beta_k)-\log(1+\gamma_k^z)\right),
\end{eqnarray}
with
\setlength{\mathindent}{0cm}
\begin{subequations}
	\begin{align}
\beta_k    & \leq \frac{|\overline{{\bf{h}}}_k{\bf{v}}_k|^2}{\sum\nolimits_{i\neq k}^K|\overline{{\bf{h}}}_k{\bf{v}}_i|^2+\overline{{\bf{h}}}_k{\bf{\Lambda}}(\overline{{\bf{h}}}_k)^H+W_{\rm{mm}}N_0}, k\in \mathcal{K},                        \\
		\gamma_k^z & \geq \frac{|\overline{{\bf{h}}^e}_z{\bf{v}}_k|^2}{\sum\nolimits_{i\neq k}^K|\overline{{\bf{h}}^e}_z{\bf{v}}_i|^2+\overline{{\bf{h}}^e}_z{\bf{\Lambda}}(\overline{{\bf{h}}^e}_z)^H+W_{\rm{mm}}N_0}, k\in \mathcal{K}, z\in \mathcal{Z}.
	\end{align}
\end{subequations}  

To this end, we define the BF matrix ${\bf{V}}_k={\bf{v}}_k{\bf{v}}_k^H$ and  ${\bf{V}}_0={\bf{v}}_0{\bf{v}}_0^H$, and (\ref{OptA}) can be reformulated as the following optimization problem
\setlength{\mathindent}{0cm}
\begin{subequations}\label{OptB}
	\begin{align}
&\underset{\left\{{\bf{V}}_0,\{{\bf{V}}_k\},{\bf{\Lambda}},\{\beta_k\},\{\gamma_k^z\}\right\}}{\rm{max}}\;\sum_{k=1}^{K}W_{\rm{mm}}\left(\log(1+\beta_k)-\log(1+\gamma_k^z)\right) \label{OptB0}\\
	{\rm{s.t.}}&\;\beta_k  \leq \frac{{\rm{Tr}}(\overline{{\bf{H}}}_k{\bf{V}}_k)}{\sum\nolimits_{i\neq k}^K{\rm{Tr}}(\overline{{\bf{H}}}_k{\bf{V}}_i)+{\rm{Tr}}(\overline{{\bf{H}}}_k{\bf{\Lambda}})+W_{\rm{mm}}N_0}, k\in \mathcal{K},\label{OptB1}\\
	&\;\gamma_k^z \!\geq\! \frac{{\rm{Tr}}(\overline{{\bf{H}}^e}_z{\bf{V}}_k)}{\sum\nolimits_{i\neq k}^K{\rm{Tr}}(\overline{{\bf{H}}^e}_z{\bf{V}}_i)\!+\!{\rm{Tr}}(\overline{{\bf{H}}^e}_z{\bf{\Lambda}})\!+\!W_{\rm{mm}}N_0}, k\!\in\! \mathcal{K},z\!\in\! \mathcal{Z},\label{OptB2}\\
	&\;\sum_{k=1}^{K}{\rm{Tr}}({\bf{V}}_k)+{\rm{Tr}}({\bf{\Lambda}})\leq P_{\rm{max}}^{\rm{BS}},\label{OptB3}\\
	&\;{\rm{Tr}}({\bf{V}}_0)\leq P_{\rm{max}}^{\rm{AC}}, \label{OptB4}\\
	&\;{\rm{Rank}}({\bf{V}}_k)=1, k\in \{0,\mathcal{K}\},\label{OptB5}\\
	&\;{\bf{V}}_k\succeq 0, k\in \{0,\mathcal{K}\},\label{OptB6}\\
	&\;(\rm{\ref{OptA1}}),\label{OptB7}
	\end{align}
\end{subequations} 
where $\overline{{\bf{H}}}_k=(\overline{{\bf{h}}}_k)^H\overline{{\bf{h}}}_k$, $\overline{{\bf{H}}^{\rm{e}}}_k=(\overline{{\bf{h}}^{\rm{e}}}_k)^H\overline{{\bf{h}}^{\rm{e}}}_k$, and  $||{\bf{F}}{\bf{v}}_k||^2={\rm{Tr}}({\bf{F}}^H{\bf{F}}{\bf{V}}_k)={\rm{Tr}}({\bf{V}}_k)$ due to ${\bf{F}}^H{\bf{F}}={\bf{I}}$. However, problem (\ref{OptB}) is still difficult to solve due to (\ref{OptB0})-(\ref{OptB2}), (\ref{OptB5}) and (\ref{OptA1}). To tackle this problem, we first transform (\ref{OptB1})  into a convex constraint. By bringing auxiliary variable $\{\varepsilon_k
\}$, we have 
\setlength{\mathindent}{1cm}
\begin{subequations}
	\begin{align}
   &\beta_k\varepsilon_k     \leq {\rm{Tr}}(\overline{{\bf{H}}}_k{\bf{V}}_k), k\in \mathcal{K}, \label{C11} \\
	 &\varepsilon_k  \geq \sum\nolimits_{i\neq k}^K{\rm{Tr}}(\overline{{\bf{H}}}_k{\bf{V}}_i)+{\rm{Tr}}(\overline{{\bf{H}}}_k{\bf{\Lambda}})+W_{\rm{mm}}N_0, k\in \mathcal{K}.\label{C12}  
	\end{align}
\end{subequations}

It is plausible that (\ref{C12}) is a convex constraint. Additionally, an upper bound of  $\beta_k\varepsilon_k$ can be obtained as follows~\cite{33SCA1}:
\begin{eqnarray}\label{C21}
\frac{\beta_k^{[n]}}{2\varepsilon_k^{[n]}}\varepsilon_k^2+\frac{\varepsilon_k^{[n]}}{2\beta_k^{[n]}}\beta_k^2\geq \beta_k\varepsilon_k,\;\;k\in \mathcal{K}, 
\end{eqnarray}
where $\beta_k^{[n]}$ and $\varepsilon_k^{[n]}$ denote the values of $\beta_k$ and $\varepsilon_k$ at the $n$th iteration, respectively. We can thus transform (\ref{C11}) into the following convex constraint
\begin{eqnarray}\label{C22}
\frac{\beta_k^{[n]}}{2\varepsilon_k^{[n]}}\varepsilon_k^2+\frac{\varepsilon_k^{[n]}}{2\beta_k^{[n]}}\beta_k^2\leq {\rm{Tr}}(\overline{{\bf{H}}}_k{\bf{V}}_k), \;\; k\in \mathcal{K}.
\end{eqnarray}

Next, we introduce the auxiliary variables $\{\zeta_k^z\}$ and $\{\mu_k^z\}$,  and  split (\ref{OptB2}) into the following constraints
\setlength{\mathindent}{0cm}
\begin{subequations}
	\begin{align}
&{\rm{Tr}}(\overline{{\bf{H}}^e}_z{\bf{V}}_k)-\gamma_k^zW_{\rm{mm}}N_0\leq \zeta_k^z, k\!\in\! \mathcal{K},z\!\in\! \mathcal{Z}, \label{C31} \\
	\;\;\;\;&\zeta_k^z  \leq (\mu_k^z)^2, k\!\in\! \mathcal{K},z\!\in\! \mathcal{Z},\label{C32}\\
	\;\;\;\;&(\mu_k^z)^2\leq \gamma_k^z\left(\sum\nolimits_{i\neq k}^K{\rm{Tr}}(\overline{{\bf{H}}^e}_z{\bf{V}}_i)\!+\!{\rm{Tr}}(\overline{{\bf{H}}^e}_z{\bf{\Lambda}})\right), k\!\in\! \mathcal{K},z\!\in\! \mathcal{Z}.\label{C33}
	\end{align}
\end{subequations}

\newcounter{mytempeqncnt}
\begin{figure*}[b]
	\hrulefill
	\setcounter{mytempeqncnt}{\value{equation}}
	\setcounter{equation}{39}
	\setlength{\mathindent}{0cm}
	\begin{align}\label{lag}
	\mathcal{F}({\bf{\Pi}}_1)=&\sum_{k=1}^{K}W_{\rm{mm}}f\left(\beta_k,\gamma_k^{z},\gamma_k^{z[n]}\right)\!+\!\psi_1\left(P_{\rm{max}}^{\rm{BS}}-\sum_{k=1}^{K}{\rm{Tr}}({\bf{V}}_k)\!-\!{\rm{Tr}}({\bf{\Lambda}})\right)\!+\!\psi_2\left(P_{\rm{max}}^{\rm{AC}}\!-\!{\rm{Tr}}({\bf{V}}_0)\right)\!+\!\sum_{k=1}^{K}\psi_3^k\left(\varepsilon_k\!-\!\sum\nolimits_{i\neq k}^K{\rm{Tr}}(\overline{{\bf{H}}}_k{\bf{V}}_i)\!-\!{\rm{Tr}}(\overline{{\bf{H}}}_k{\bf{\Lambda}})\!-\!W_{\rm{mm}}N_0\right)\nonumber\\
	+&\sum_{k=1}^{K}\psi_4^k\left( {\rm{Tr}}(\overline{{\bf{H}}}_k{\bf{V}}_k)\!-\!\frac{\beta_k^{[n]}}{2\varepsilon_k^{[n]}}\varepsilon_k^2\!-\!\frac{\varepsilon_k^{[n]}}{2\beta_k^{[n]}}\beta_k^2\right)+\sum_{z=1}^{Z}\sum_{k=1}^{K}\psi_5^{z,k}\left(\zeta_k^z\!+\!\gamma_k^zW_{\rm{mm}}N_0\!-\!{\rm{Tr}}(\overline{{\bf{H}}^e}_z{\bf{V}}_k)\right)\!+\!\sum_{z=1}^{Z}\sum_{k=1}^{K}\psi_6^{z,k}\left(\sum\nolimits_{i\neq k}^K{\rm{Tr}}(\overline{{\bf{H}}^e}_z{\bf{V}}_i)\!+\!{\rm{Tr}}(\overline{{\bf{H}}^e}_z{\bf{\Lambda}})\right)\nonumber\\
	+&\sum_{k=1}^{K}\psi_7^{k}\left(\theta_k+\tau_kW_{\rm{mm}}N_0-{\rm{Tr}}(\overline{{\bf{H}}}_k{\bf{V}}_k)\right)+\sum_{k=1}^{K}\psi_8^{k}\left(\sum\nolimits_{i\neq k}^K{\rm{Tr}}(\overline{{\bf{H}}}_k{\bf{V}}_i)\!+\!{\rm{Tr}}(\overline{{\bf{H}}}_k{\bf{\Lambda}})\right)+\sum_{k=0}^{K}{\rm{Tr}}({{\bf{\Omega}}_k\bf{V}}_k)\nonumber\\
	+&\sum_{l=1}^{L}\psi_9^{l}\left({\rm{Tr}}({\bf{G}}_l{\bf{V}}_0)-W_{\rm{mc}}N_0(e^{\omega/\eta}-1)\right)+\psi.
	\end{align}
	where
	${\bf{\Pi}}_1=\left\{{\bf{V}}_0,\{{\bf{V}}_k\},{\bf{\Lambda}},\{\varepsilon_k\},\{\lambda_k\},\{\theta_k\},\{\zeta_k^z\},\{\mu_k^z\},\{\tau_k\},\omega,\{\beta_k\},\{\gamma_k^z\},\psi_1,\psi_2, \{\psi_3^k\},\{\psi_4^k\},\{\psi_5^{z,k}\},\{\psi_6^{z,k}\},\{\psi_7^k\},\{\psi_8^k\},\{\psi_9^{l}\},\{{\bf{\Omega}}_k\},\psi \right\}$ and $\psi$ denotes other Lagrangian terms that do not affect our analysis. 
	\setcounter{equation}{\value{mytempeqncnt}}
\end{figure*}

It can be readily seen that (\ref{C31}) is a convex constraint.  As for (\ref{C32}), by the first-order Taylor series expansion, the quadratic term $(\mu_k^z)^2$ can be approximated as 
\begin{eqnarray}\label{C410}
	\left(\mu_k^z\right)^2\approx 2\mu_k^{z[n]}\mu_k^z-\left(\mu_k^{z[n]}\right)^2,k\!\in\! \mathcal{K},z\!\in\! \mathcal{Z},
\end{eqnarray} 
and (\ref{C32}) can be transformed into the following convex constraint 
\begin{eqnarray}\label{C41}
\zeta_k^z\leq 2\mu_k^{z[n]}\mu_k^z-\left(\mu_k^{z[n]}\right)^2,k\!\in\! \mathcal{K},z\!\in\! \mathcal{Z},
\end{eqnarray}
where $\mu_k^{z[n]}$ denotes the value of $\mu_k^z$ at the $n$th iteration.
As for the non-linear constraint~(\ref{C33}), we can transform it into the following convex linear matrix inequality (LMI) constraint as
\begin{eqnarray}\label{C51A}
\left[ \begin{array}{ccc}
\gamma_k^z & \mu_k^z \\
\mu_k^z & \sum\nolimits_{i\neq k}^K{\rm{Tr}}(\overline{{\bf{H}}^e}_z{\bf{V}}_i)\!+\!{\rm{Tr}}(\overline{{\bf{H}}^e}_z{\bf{\Lambda}}) 
\end{array} 
\right ]\succeq {\bf{0}},\;\;k\!\in\! \mathcal{K},z\!\in\! \mathcal{Z}.\label{C51}
\end{eqnarray}

Finally, upon introducing auxiliary variables $\{\tau_k\}$ and $\omega$, (\ref{OptB}) can be transformed into the following optimization problem
\begin{subequations}\label{OptC}
	\begin{align}
	&\underset{\left\{{\bf{V}}_0,\{\varepsilon_k\},\{\zeta_k^z\},\{\mu_k^z\},\{\tau_k\},\omega,\{{\bf{V}}_k\},{\bf{\Lambda}},\{\beta_k\},\{\gamma_k^z\}\right\}}{\rm{max}}\;\sum_{k=1}^{K}W_{\rm{mm}}\left(\log(1\!+\!\beta_k)\!-\!\log(1\!+\!\gamma_k^z)\right) \label{OptC0}\\
	&{\rm{s.t.}}\;\;\tau_k  \geq \frac{{\rm{Tr}}(\overline{{\bf{H}}}_k{\bf{V}}_k)}{\sum\nolimits_{i\neq k}^K{\rm{Tr}}(\overline{{\bf{H}}}_k{\bf{V}}_i)+{\rm{Tr}}(\overline{{\bf{H}}}_k{\bf{\Lambda}})+W_{\rm{mm}}N_0}, k\in \mathcal{K},\label{OptC1}\\
	&\;\;\omega\leq \eta\log\left(1+\frac{{\rm{Tr}}({\bf{G}}_l{\bf{V}}_0)}{W_{\rm{mc}}N_0}\right), l\in\mathcal{L},\label{OptC2}\\
	&\;\;\omega-\sum\nolimits_{k=1}^{K}\log(1+\tau_k)\geq 0,\label{OptC3}\\
	&\;\;{\rm(\ref{OptB3})\!-\!(\ref{OptB6}), (\ref{C12}), (\ref{C22}), (\ref{C31}), (\ref{C41}), (\ref{C51})},
	\end{align}
\end{subequations} 
where ${\bf{G}}_l={\bf{g}}_l^H{\bf{g}}_l$, and $\eta=W_{\rm{mc}}/W_{\rm{mm}}$.

Since $\log(1+\beta_k)$ and $\log(1+\gamma_k^z)$ are convex functions, the OF (\ref{OptC0}) is constituted by a difference of convex~(DC) program~\cite{34DCP}, and usually CCCP is used to solve the DC program. In fact, the main idea of the CCCP is to transform (\ref{OptC0}) into a convex function by approximation, and then the approximated convex problem is iteratively solved until the result converges. Based on this, we consider to use the first-order Taylor approximation for $\log(1+\gamma_k^z)$, namely
\begin{eqnarray}
	\log\left(1+\gamma_k^z\right)\approx \log\left(1+\gamma_k^{z[n]}\right)+\frac{\gamma_k^z-\gamma_k^{z[n]}}{1+\gamma_k^{z[n]}},
\end{eqnarray}
where $\gamma_k^{z[n]}$ denotes the value of $\gamma_k^z$ at the $n$th iteration. Then, the OF can be transformed into the following convex function
\begin{eqnarray}
	f\left(\beta_k,\gamma_k^{z},\gamma_k^{z[n]}\right)=\log\left(1+\beta_k\right)-\log\left(1+\gamma_k^{z[n]}\right)-\frac{\gamma_k^z-\gamma_k^{z[n]}}{1+\gamma_k^{z[n]}}.
\end{eqnarray}

Additionally, one can observe that constraint (\ref{OptC3}) is also a DC program. Similarly, we can transform it into the following convex constraint
\begin{eqnarray}
	\omega-\sum_{k=1}^{K}\log\left(1+\tau_k^{[n]}\right)-\frac{\tau_k-\tau_k^{[n]}}{1+\tau_k^{[n]}}\geq 0.
\end{eqnarray}

The non-convex constraint  (\ref{OptC1}) is yet to be dealt with.  Similar to (\ref{OptB2}), by bringing auxiliary variable $\{\theta_k\}$ and $\{\lambda_k\}$, (\ref{OptC1}) can be split into the following constraints
\begin{subequations}
	\begin{align}
&{\rm{Tr}}(\overline{{\bf{H}}}_k{\bf{V}}_k)-\tau_kW_{\rm{mm}}N_0\leq \theta_k, k\!\in\! \mathcal{K}, \label{C61} \\
	&\theta_k  \leq (\lambda_k)^2, k\!\in\! \mathcal{K},\label{C62}\\
	&(\lambda_k)^2\leq \tau_k\left(\sum\nolimits_{i\neq k}^K{\rm{Tr}}(\overline{{\bf{H}}}_k{\bf{V}}_i)\!+\!{\rm{Tr}}(\overline{{\bf{H}}}_k{\bf{\Lambda}})\right), k\!\in\! \mathcal{K}.\label{C63}
	\end{align}
\end{subequations}

According to (\ref{C410})-(\ref{C51A}), the non-convex constraints (\ref{C62}) and (\ref{C63}) can be formulated the following convex ones 
\begin{eqnarray}\label{C71}
2\lambda_k^{[n]}\lambda_k-\left(\lambda_k^{[n]}\right)^2-\theta_k\geq 0,k\!\in\! \mathcal{K},
\end{eqnarray}
where $\lambda_k^{[n]}$ represents the value of $\lambda_k$ at the $n$th iteration, and 
\begin{eqnarray}\label{C81}
\left[ \begin{array}{ccc}
\tau_k & \lambda_k \\
\lambda_k & \sum\nolimits_{i\neq k}^K{\rm{Tr}}(\overline{{\bf{H}}}_k{\bf{V}}_i)\!+\!{\rm{Tr}}(\overline{{\bf{H}}}_k{\bf{\Lambda}}) 
\end{array} 
\right ]\succeq {\bf{0}},\;\;k\!\in\! \mathcal{K}.
\end{eqnarray}  

Finally, the SRM  problem can be transformed into
\begin{subequations}\label{OptD}
	\begin{align}
		&\underset{\left\{{\bf{V}}_0,\{\varepsilon_k\},\{\lambda_k\},\{\theta_k\},\{\zeta_k^z\},\{\mu_k^z\},\{\tau_k\},\omega,\{{\bf{V}}_k\},{\bf{\Lambda}},\{\beta_k\},\{\gamma_k^z\}\right\}}{\rm{max}}\;\sum_{k=1}^{K}W_{\rm{mm}}f\left(\beta_k,\gamma_k^{z},\gamma_k^{z[n]}\right) \label{OptD0}\\
	&{\rm{s.t.}}\;\;{\rm(\ref{OptB3})-(\ref{OptB6}), (\ref{C12}), (\ref{C22}), (\ref{C31}), (\ref{C41}), (\ref{C51}),(\ref{OptC2}) }, \nonumber\\
	&\;\;\;\;\;\;{\rm(\ref{C61}), (\ref{C71}), (\ref{C81})}.
	\end{align}
\end{subequations}
 
It is plausible that only the rank-one constraint (\ref{OptB5}) is non-convex in (\ref{OptD}).  By SDP relaxation (i.e., removing the rank-one constraint), (\ref{OptD}) will become a convex optimization problem and can be solved by standard convex optimization techniques, such as the interior-point method of~\cite{35CVX}. 
Finally, to obtain the solution of optimization problem~(\ref{OptA}), we need to iteratively solve~(\ref{OptD}). Specifically, we first initialize the feasible solution $\{\beta_k^{[n]}\}$, $\{\varepsilon_k^{[n]}\}$, $\{\mu_k^{z[n]}\}$, $\{\gamma_k^{z[n]}\}$, $\{\tau_k^{[n]}\}$ as well as $\{\lambda_k^{[n]}\}$, and the optimal solution of (\ref{OptD}) can be obtained by a classic convex optimization algorithm. Then, $\{\beta_k^{[n+1]}\}$, $\{\varepsilon_k^{[n+1]}\}$, $\{\mu_k^{z[n+1]}\}$, $\{\gamma_k^{z[n+1]}\}$, $\{\tau_k^{[n+1]}\}$ and $\{\lambda_k^{[n]}\}$ are updated according to the solution obtained  at previous iteration, and problem~(\ref{OptD}) is resolved until the results converge or the iteration index reaches its maximum value. Additionally, since~(\ref{OptD}) without rank-one constraint is a convex optimization problem, iteratively updating all variables will increase or at least maintain the value of the OF in~(\ref{OptD})~\cite{36Zhang},~\cite{add3}.  Given the limited transmit power, the value of the OF  should be monotonically non-decreasing sequence with an upper bound, which converges to a stationary solution that is at least locally optimal. We summarize the above iterative scheme in Algorithm~\ref{algorithm2}.

Now, we study the rank relaxation problem. To analyze the rank characteristic of the solution, we first define the Lagrangian function of the relaxed version of (\ref{OptD}) that can be expressed as ~(\ref{lag}) at the bottom of this page. Then, we have the following~theorems.

 \setcounter{equation}{40}
 \begin{theorem}\label{theorem1}
 When $\psi_2^\ast>0$, let ${\bf{V}}_0^\ast$ denote an optimal solution of~(\ref{OptD}), it must hold that
 \begin{eqnarray}
 	{\rm{rank}}({{\bf{V}}}_0^\ast)\leq L,
 \end{eqnarray}
 where $\psi_2^\ast$ represents the optimal Lagrange multiplier for the dual problem of~(\ref{OptD}), while there always exists a ${\bf{V}}_0^\ast$ so that 
  \begin{eqnarray}
 {\rm{rank}}({{\bf{V}}}_0^\ast)\leq \sqrt{L}.
 \end{eqnarray}
 
 Moreover, if there exists any $l$ so that  $\|{\bf{g}}_{l}\|<\|{\bf{g}}_{l'}\|(l'\in\{1,\dots,l-1,l+1,\dots,L\})$, we have 
   \begin{eqnarray}
 {\rm{rank}}({{\bf{V}}}_0^\ast)=1.
 \end{eqnarray}
 \end{theorem}
\begin{proof}
Refer to Appendix~\ref{appendixA}.
\end{proof}
\begin{theorem}\label{theorem2}
	When $\psi_1^\ast>0, \psi_3^{i\ast}-\psi_8^{i\ast}\geq 0\;(i\neq k), \psi_5^{{z,k}\ast}-\sum_{i\neq k}^{K}\psi_6^{z,i\ast}\geq 0$, we have 
  \begin{eqnarray}
{\rm{rank}}({{\bf{V}}}_k^\ast)=1, k\in{\mathcal{K}},
\end{eqnarray}
where $\psi_1^\ast, \psi_3^{i\ast}, \psi_8^{i\ast}, \psi_5^{z,k\ast}, \psi_6^{z,i\ast}$ denote the optimal Lagrange multipliers for the dual problem of~(\ref{OptD}).
\end{theorem}
\begin{proof}
	Refer to Appendix~\ref{appendixB}.
\end{proof}

It should be stressed that Theorems~\ref{theorem1} and~\ref{theorem2} provide only the sufficient conditions for the optimality of the BF. During simulation experience, we found that problem~(\ref{OptD}) still yields the rank-one solutions even though the above conditions do not hold. However, we still have to consider the following problem: If $\{{{\bf{V}}}_0^\ast,\{{{\bf{V}}}_k^\ast\}, {\bf{\Lambda}}^\ast\}$ denotes the optimal solution of problem~(\ref{OptD}) and there exists ${\rm{rank}}({{\bf{V}}}_k^\ast)>1$ or  ${\rm{rank}}({{\bf{V}}}_0^\ast)>1$, how can the rank-one solution be obtained? To solve this problem, we consider the following two cases:
\begin{itemize}
	\item When ${\rm{rank}}({{\bf{V}}}_k^\ast)>1$, we can construct a feasible  solution $\{{\hat{\bf{V}}}_0^\ast,\{{\hat{\bf{V}}}_k^\ast\}, \hat{\bf{\Lambda}}^\ast\}$ for ${\rm{rank}}({{\bf{V}}}_k^\ast)=1$ such that it can achieve the same objective value with $\{{{\bf{V}}}_0^\ast,\{{{\bf{V}}}_k^\ast\}, {\bf{\Lambda}}^\ast\}$. Refer to Appendix~\ref{appendixC} for details.
	\item When ${\rm{rank}}({{\bf{V}}}_0^\ast)>1$, we apply a randomization technique to obtain a rank-one ${{\bf{V}}}_0^\ast$. The details can be found in Appendix~\ref{appendixD}.
	
\end{itemize}

\begin{algorithm}[t]
	{\caption{The Proposed Iterative Algorithm.}
		\label{algorithm2}
		{\bf{Initialize}} $\{\beta_k^{[n]}\}$, $\{\varepsilon_k^{[n]}\}$, $\{\mu_k^{z[n]}\}$, $\{\gamma_k^{z[n]}\}$, $\{\tau_k^{[n]}\}$, $\{\lambda_k^{[n]}\}$, $n=0$, the maximum iteration index $T_{\rm{max}}$.\\
		\Repeat{$n=T_{\rm{max}}$ {\rm{or} Convergence}}{
			Update $n\leftarrow n+1$.\\
			Solve the optimization problem (\ref{OptD}) without the rank-one constraint and obtain the optimal solution ${\bf{V}}_0^{[n]},\{\varepsilon_k^{[n]}\},\{\lambda_k^{[n]}\},\{\theta_k^{[n]}\},\{\zeta_k^{z[n]}\},\{\mu_k^{z[n]}\},\{\tau_k^{[n]}\},\omega^{[n]},\{{\bf{V}}_k^{[n]}\},{\bf{\Lambda}}^{[n]}$,\\$\{\beta_k^{[n]}\},\{\gamma_k^{z[n]}\}$.}}
\end{algorithm} 

\subsection{Problem Formulation and Solution Under Per-BS Transmit Power Constraint}
In the previous subsection, the total  transmit power constraint of $L$ BSs is considered. Although the power can be allocated more flexibly to the BSs under total transmit power constraint,  each BS is usually prohibited to transmit at high power due to practicality concerns. Therefore, considering per-BS transmit power constraint may be more practical.  To this end, we define ${\bf{B}}$ as
	\begin{eqnarray}\label{powerconstraint}
	{\bf{B}}_l\triangleq {\rm{Diag}}(\underbrace{0,\ldots,0}_{(l-1)},1,\underbrace{0,\ldots,0}_{(L-l)}), l\in{\mathcal{L}},
	\end{eqnarray}
and per-BS power constraint (\ref{OptB3}) can be written~as
\begin{eqnarray}\label{Per-BS}
	\sum_{k=1}^{K}{\rm{Tr}}({\bf{V}}_k{\bf{B}}_l)+{\rm{Tr}}({\bf{\Lambda}}{\bf{B}}_l)\leq P_{{\rm{max}},l}^{{\rm{BS}}},l\in{\mathcal{L}},
\end{eqnarray}
where $P_{{\rm{max}},l}^{{\rm{BS}}}$ denotes the maximum transmit power of the $l$th BS. Then, we formulate the following SRM  problem under per-BS transmit power constraint as follows:
\begin{subequations}\label{OptAA}
	\begin{align}
	\;\;\;\;\;\;\;\;\;\;\;\;\;\;\;&\underset{\left\{{\bf{V}}_0,\{{\bf{V}}_k\},{\bf{\Lambda}}\right\}}{\rm{max}}\;\sum_{k=1}^{K}\left[R_k^{\rm{AC}}-\underset{z\in{\mathcal{Z}}}{\max}\left\{R_{k,z}^{\rm{EV}}\right\}\right]^+ \label{OptAA0}\\
	{\rm{s.t.}}&\;\; {\rm(\ref{OptA1})-(\ref{OptA3}),(\ref{Per-BS})}.
	\end{align}
\end{subequations} 

Since (\ref{Per-BS}) is a convex constraint, we can adopt the scheme proposed in Section III.~B to deal with problem (\ref{OptAA}), and finally formulate the following convex optimization problem without a rank-one constraint
\begin{subequations}\label{OptE}
	\begin{align}
&\underset{\left\{{\bf{V}}_0,\{\varepsilon_k\},\{\lambda_k\},\{\theta_k\},\{\zeta_k^z\},\{\mu_k^z\},\{\tau_k\},\omega,\{{\bf{V}}_k\},{\bf{\Lambda}},\{\beta_k\},\{\gamma_k^z\}\right\}}{\rm{max}}\;\sum_{k=1}^{K}W_{\rm{mm}}f\left(\beta_k,\gamma_k^{z},\gamma_k^{z[n]}\right) \label{OptE0}\\
	&{\rm{s.t.}}\;\;{\rm(\ref{OptB4}), (\ref{OptB6}), (\ref{C12}), (\ref{C22}), (\ref{C31}), (\ref{C41}), (\ref{OptC2}),(\ref{C51})},{\rm(\ref{C61})}, \nonumber\\
	&\;\;\;\;\;\; {(\ref{C71}), (\ref{C81}), (\ref{Per-BS})}.
	\end{align}
\end{subequations}

As a result, we can iteratively solve the above convex optimization problem to obtain the solution of the original problem~(\ref{OptAA}) using a process similar  to~Algorithm~\ref{algorithm2}.  Additionally, we can obtain the similar theorems with Theorems~\ref{theorem1} and~\ref{algorithm2}, and the reconstruction of rank-one solution can refer to Appendices~\ref{appendixA}-\ref{appendixD}.
\subsection{Complexity Analysis}
We first analyze the complexity of solving~(\ref{OptD}). Given an iteration accuracy $\epsilon$, the number of iterations is on the order $\sqrt{\Delta}\ln(1/\epsilon)$, where $\Delta=4KZ\!+\!KL\!+\!8K\!+\!N\!+\!1$ denotes the threshold parameters related to the constraints~\cite{43Com}. Additionally, (\ref{OptD}) includes $(2KZ+4K+1)$ linear constraints, $(KZ+K)$ two-dimensional LMI constraints, one $N$-dimensional  LMI constraint, $K$ $L$-dimensional  LMI  constraints and $K$ second order cone constraints. Therefore, the complexity of solving~(\ref{OptD}) is on the order of 
\begin{eqnarray}
	{\mathcal{O}}\left(\xi\sqrt{\Delta}\ln(1/\epsilon)(\xi_1+\xi_2\xi+\xi_3+\xi^2)\right),
\end{eqnarray} 
where $\xi={\mathcal{O}}(N^2+KL^2)$ and $N^2+KL^2$ denotes the number of decision variables,  $\xi_1=KL^3+N^3+10KZ+12K+1$, $\xi_2=KL^2+N^2+6KZ+8K+1$, $\xi_3=K(L+2)^2$. Additionally, one can observe that the only difference  between problems~(\ref{OptE}) and~(\ref{OptD}) is (\ref{Per-BS}). Since (\ref{Per-BS}) is also a linear constraint, the complexity of solving~(\ref{OptE}) is on the order of 
\begin{eqnarray}
{\mathcal{O}}\left(\xi\sqrt{\Delta}\ln(1/\epsilon)(\hat{\xi}_1+\hat{\xi}_2\xi+\xi_3+\xi^2)\right),
\end{eqnarray} 
where $\hat{\xi}_1\!=\!KL^3\!+\!N^3\!+\!10KZ\!+\!12K\!+\!L$, $\hat{\xi}_2\!=\!KL^2\!+\!N^2\!+\!6KZ\!+\!8K\!+\!L$.
\section{Extension To The Imperfect CSI for Eavesdropper Links}
In the previous section, we consider the SRM problem under the assumption that the Eves' CSIs are perfectly known by the BSs. In fact, since the Eves are usually passive, the BSs cannot obtain their perfect CSIs~\cite{37Imcsi},~\cite{38Imcsi}. Therefore,  we assume that the BSs own imperfect CSIs for the Eves links,~i.e.,
\begin{eqnarray}
	{\bf{h}}_z^{\rm{e}}={\hat{\bf{h}}}_z^{\rm{e}}+\bigtriangleup{\bf{h}}_z^{\rm{e}},\;\;z\in{\mathcal{Z}},
\end{eqnarray}
where ${\bf{h}}_z^{\rm{e}}$ is the actual channel vector from $L$ BSs to the $z$th Eve, ${\hat{\bf{h}}}_z^{\rm{e}}$ denotes the estimated ${{\bf{h}}}_z^{\rm{e}}$, and $\bigtriangleup{\bf{h}}_z^{\rm{e}}$ represents the associated CSI error. Here, we assume that the estimation error $\bigtriangleup{\bf{h}}_z^{\rm{e}}$ is bounded and deterministic, which is defined as
\begin{eqnarray}
	||\bigtriangleup{\bf{h}}_z^{\rm{e}}||^2\leq \sigma_z ||\hat{\bf{h}}_z^{\rm{e}}||^2, z\in{\mathcal{Z}},
\end{eqnarray}
where $\sigma_z\geq 0$ is the error ratio. Hence, the received signal of the $k$th user at the $z$th Eve can be rewritten as
\begin{eqnarray}
\begin{aligned}
y_z^{\rm{EV}}=({\hat{\bf{h}}}_z^{\rm{e}}+\bigtriangleup{\bf{h}}_z^{\rm{e}}){\bf{F}}\left(\sum\nolimits_{i=1}^K{\bf{v}}_ix_i+{\bf{q}}\right)+n_z.
\end{aligned}
\end{eqnarray}
{The $z$th Eve has an achievable rate at the $k$th user represented as}
\begin{eqnarray}
\hat{R}_{k,z}^{\rm{EV}}\!=\!W_{\rm{mm}}\log\left(1\!+\!\frac{|({\hat{\bf{h}}}_z^{\rm{e}}+\bigtriangleup{\bf{h}}_z^{\rm{e}}){\bf{F}}{\bf{v}}_k|^2}{\Sigma_{k,z}+\!W_{\rm{mm}}N_0}\right).
\end{eqnarray} 
where $\Sigma_{k,z}=\sum\nolimits_{i\neq k}^K|({\hat{\bf{h}}}_z^{\rm{e}}\!+\!\bigtriangleup{\bf{h}}_z^{\rm{e}}){\bf{F}}{\bf{v}}_i|^2\!+\!|({\hat{\bf{h}}}_z^{\rm{e}}\!+\!\bigtriangleup{\bf{h}}_z^{\rm{e}}){\bf{F}}{\bf{q}}|^2$.
Then, we can formulate the SRM problem as follows
\begin{subequations}\label{OptF}
	\begin{align}
	\;\;\;\;\;\;\;\;\;\;\;\;&\underset{\left\{{\bf{v}}_0,\{{\bf{v}}_k\},{\bf{\Lambda}}\right\}}{\rm{max}}\;\sum_{k=1}^{K}\left[R_k^{\rm{AC}}-\underset{z\in{\mathcal{Z}}}{\max}\left\{\hat{R}_{k,z}^{\rm{EV}}\right\}\right]^+ \label{OptF0}\\
	{\rm{s.t.}}&\;\; {\bigtriangleup{\bf{h}}_z^{\rm{e}}}({\bigtriangleup{{\bf{h}}}_z^{\rm{e}}})^H\leq \sigma_z \left({\hat{\bf{h}}_z^{\rm{e}}}({\hat{\bf{h}}_z^{\rm{e}}})^H\right),z\in{\mathcal{Z}},\label{OptF1}\\
	&\;\;{\rm{(\ref{OptA1})-(\ref{OptA3})}}.
	\end{align}
\end{subequations}

Observe that problem (\ref{OptF}) is similar to (\ref{OptA}) except for the wiretap rate $\hat{R}_{k,z}^{\rm{EV}}$ and constraint (\ref{OptF1}). Therefore, we first deal with these two terms, while the other terms can be transformed into convex ones using the similar methods to that of Section III.~B. By introducing auxiliary variable $\{\hat{\gamma}_k^z\}$, we reformulate the optimization problem of (\ref{OptF}) as
 \begin{subequations}\label{OptG}
 	\begin{align}
 		&\underset{\left\{{\bf{V}}_0,\{{\bf{V}}_k\},{\bf{\Lambda}},\{{\hat{\gamma}}_k^z\},\{{\beta}_k\}\right\}}{\rm{max}}\;\sum_{k=1}^{K}W_{\rm{mm}}\left(\log(1+\beta_k)-\log(1+\hat{\gamma}_k^z)\right) \label{OptG0}\\
 	{\rm{s.t.}}&\;\;\hat{\gamma}_k^z \!\geq\! \frac{({\hat{\bf{h}}}_z^{\rm{e}}+\bigtriangleup{\bf{h}}_z^{\rm{e}}){\bf{F}}{\bf{V}}_k{\bf{F}}^H\left({\hat{\bf{h}}}_z^{\rm{e}}+\bigtriangleup{\bf{h}}_z^{\rm{e}}\right)^H}{\Xi_k^z}, k\!\in\! \mathcal{K},z\!\in\! \mathcal{Z},\label{OptG1}\\
 	&\;\;{\rm{(\ref{OptA1}), (\ref{OptB1}), (\ref{OptB3})-(\ref{OptB6}), (\ref{OptF0})}},
 	\end{align}
 \end{subequations} 
where $\Xi_k^z=\sum\nolimits_{i\neq k}^K({\hat{\bf{h}}}_z^{\rm{e}}+\bigtriangleup{\bf{h}}_z^{\rm{e}}){\bf{F}}{\bf{V}}_i{\bf{F}}^H\left({\hat{\bf{h}}}_z^{\rm{e}}+\bigtriangleup{\bf{h}}_z^{\rm{e}}\right)^H\!+\!({\hat{\bf{h}}}_z^{\rm{e}}+\bigtriangleup{\bf{h}}_z^{\rm{e}}){\bf{F}}{\bf{\Lambda}}{\bf{F}}^H\left({\hat{\bf{h}}}_z^{\rm{e}}+\bigtriangleup{\bf{h}}_z^{\rm{e}}\right)^H\!+\!W_{\rm{mm}}N_0$.

Next, we introduce auxiliary variables $\{\hat{\zeta}_k^z\}$, $\{\hat{\mu}_k^z\}$, $\{\chi_k^z\}$ and split the constraint (\ref{OptG1}) into the following~ones
\begin{subequations}
	\begin{align}
		\;\;\;\;\;\;&({\hat{\bf{h}}}_z^{\rm{e}}+\bigtriangleup{\bf{h}}_z^{\rm{e}}){\bf{F}}{\bf{V}}_k{\bf{F}}^H\left({\hat{\bf{h}}}_z^{\rm{e}}+\bigtriangleup{\bf{h}}_z^{\rm{e}}\right)^H\leq \hat{\zeta}_k^z, k\!\in\! \mathcal{K},z\!\in\! \mathcal{Z}, \label{C91} \\
	&\hat{\zeta}_k^z  \leq (\hat{\mu}_k^z)^2, k\!\in\! \mathcal{K},z\!\in\! \mathcal{Z},\label{C92}\\
	&(\hat{\mu}_k^z)^2\leq \hat{\gamma}_k^z\chi_k^z, k\!\in\! \mathcal{K},z\!\in\! \mathcal{Z},\label{C93}\\
	&\chi_k^z\leq \Xi_k^z, k\!\in\! \mathcal{K},z\!\in\! \mathcal{Z}.\label{C94}
	\end{align}
\end{subequations}

It is observed that (\ref{C92}) can be transformed into the following convex constraint
\begin{eqnarray}\label{C101}
\hat{\zeta}_k^z\leq 2\hat{\mu}_k^{z[n]}\hat{\mu}_k^z-\left(\hat{\mu}_k^{z[n]}\right)^2,k\!\in\! \mathcal{K},z\!\in\! \mathcal{Z},
\end{eqnarray}
where $\hat{\mu}_k^{z[n]}$ denotes the value of $\hat{\mu}_k^z$ at the $n$th iteration. Additionally, (\ref{C93}) can be reformulated into the following LMI constraint 
\begin{eqnarray}\label{C111}
\left[ \begin{array}{ccc}
\hat{\gamma}_k^z & \hat{\mu}_k^z \\
\hat{\mu}_k^z & \chi_k^z
\end{array} 
\right ]\succeq {\bf{0}},\;\;k\!\in\! \mathcal{K},z\!\in\! \mathcal{Z}.
\end{eqnarray}  

\begin{figure*}[b]
	\hrulefill
	\setcounter{mytempeqncnt}{\value{equation}}
	\setcounter{equation}{63}
	\setlength{\mathindent}{0cm}
	\begin{eqnarray}\label{C15}
	\begin{aligned}
	\Xi_k^z=&\sum\nolimits_{i\neq k}^K({\hat{\bf{h}}}_z^{\rm{e}}+\bigtriangleup{\bf{h}}_z^{\rm{e}}){\bf{F}}{\bf{V}}_i{\bf{F}}^H\left({\hat{\bf{h}}}_z^{\rm{e}}+\bigtriangleup{\bf{h}}_z^{\rm{e}}\right)^H\!+\!({\hat{\bf{h}}}_z^{\rm{e}}+\bigtriangleup{\bf{h}}_z^{\rm{e}}){\bf{F}}{\bf{\Lambda}}{\bf{F}}^H\left({\hat{\bf{h}}}_z^{\rm{e}}+\bigtriangleup{\bf{h}}_z^{\rm{e}}\right)^H\!+\!W_{\rm{mm}}N_0\\
	=&\bigtriangleup{\bf{h}}_z^{\rm{e}}\left(\sum\nolimits_{i\neq k}^K{\bf{F}}{\bf{V}}_i{\bf{F}}^H\!+\!{\bf{F}}{\bf{\Lambda}}{\bf{F}}^H\right)(\bigtriangleup{\bf{h}}_z^{\rm{e}})^H\!+\!2{\rm{Re}}
	\left\{{\hat{\bf{h}}}_z^{\rm{e}}\left(\sum\nolimits_{i\neq k}^K{\bf{F}}{\bf{V}}_i{\bf{F}}^H\!+\!{\bf{F}}{\bf{\Lambda}}{\bf{F}}^H\right)(\bigtriangleup{\bf{h}}_z^{\rm{e}})^H\right\}\!+\!{\hat{\bf{h}}}_z^{\rm{e}}\left(\sum\nolimits_{i\neq k}^K{\bf{F}}{\bf{V}}_i{\bf{F}}^H\!+\!{\bf{F}}{\bf{\Lambda}}{\bf{F}}^H\right)(\hat{\bf{h}}_z^{\rm{e}})^H\!+\!W_{\rm{mm}}N_0.
	\end{aligned}
	\end{eqnarray}

	\begin{eqnarray}\label{C16}
	\left[ \begin{array}{ccc}
	\upsilon_z{\bf{I}}\!+\!\left(\sum\nolimits_{i\neq k}^K{\bf{F}}{\bf{V}}_i{\bf{F}}^H\!+\!{\bf{F}}{\bf{\Lambda}}{\bf{F}}^H\right) & \left({\hat{\bf{h}}}_z^{\rm{e}}\left(\sum\nolimits_{i\neq k}^K{\bf{F}}{\bf{V}}_i{\bf{F}}^H\!+\!{\bf{F}}{\bf{\Lambda}}{\bf{F}}^H\right)\right)^H \\
	{\hat{\bf{h}}}_z^{\rm{e}}\left(\sum\nolimits_{i\neq k}^K{\bf{F}}{\bf{V}}_i{\bf{F}}^H\!+\!{\bf{F}}{\bf{\Lambda}}{\bf{F}}^H\right) &-\upsilon_z\sigma_z {\hat{\bf{h}}_z^{\rm{e}}}({\hat{\bf{h}}_z^{\rm{e}}})^H\!+\!{\hat{\bf{h}}}_z^{\rm{e}}\left(\sum\nolimits_{i\neq k}^K{\bf{F}}{\bf{V}}_i{\bf{F}}^H\!+\!{\bf{F}}{\bf{\Lambda}}{\bf{F}}^H\right)(\hat{\bf{h}}_z^{\rm{e}})^H\!+\!W_{\rm{mm}}N_0-{\chi_z^k}
	\end{array} 
	\right ]\succeq {\bf{0}}.
	\end{eqnarray}  
	
	\setcounter{equation}{66}
	
	\begin{align}\label{lagA}
	\mathcal{F}({\bf{\Pi}}_2)=&\sum_{k=1}^{K}W_{\rm{mm}}\hat{f}\left(\beta_k,\hat{\gamma}_k^{z},\hat{\gamma}_k^{z[n]}\right)\!+\psi_1\left(P_{\rm{max}}^{\rm{BS}}-\sum_{k=1}^{K}{\rm{Tr}}({\bf{V}}_k)\!-\!{\rm{Tr}}({\bf{\Lambda}})\right)\!+\!\psi_2\left(P_{\rm{max}}^{\rm{AC}}\!-\!{\rm{Tr}}({\bf{V}}_0)\right)\!+\!\sum_{k=1}^{K}\psi_3^k\left(\varepsilon_k\!-\!\sum\nolimits_{i\neq k}^K{\rm{Tr}}(\overline{{\bf{H}}}_k{\bf{V}}_i)\!-\!{\rm{Tr}}(\overline{{\bf{H}}}_k{\bf{\Lambda}})\!-\!W_{\rm{mm}}N_0\right)\nonumber\\
	+&\sum_{k=1}^{K}\psi_4^k\left( {\rm{Tr}}(\overline{{\bf{H}}}_k{\bf{V}}_k)\!-\!\frac{\beta_k^{[n]}}{2\varepsilon_k^{[n]}}\varepsilon_k^2\!-\!\frac{\varepsilon_k^{[n]}}{2\beta_k^{[n]}}\beta_k^2\right)+\sum_{k=1}^{K}\psi_7^{k}\left(\theta_k+\tau_kW_{\rm{mm}}N_0-{\rm{Tr}}(\overline{{\bf{H}}}_k{\bf{V}}_k)\right)+\sum_{k=1}^{K}\psi_8^{k}\left(\sum\nolimits_{i\neq k}^K{\rm{Tr}}(\overline{{\bf{H}}}_k{\bf{V}}_i)\!+\!{\rm{Tr}}(\overline{{\bf{H}}}_k{\bf{\Lambda}})\right)\nonumber\\
	+&\sum_{l=1}^{L}\psi_9^{l}\left({\rm{Tr}}({\bf{G}}_l{\bf{V}}_0)-W_{\rm{mc}}N_0(e^{\omega/\eta}-1)\right)+\sum_{z=1}^{Z}\sum_{k=1}^{K}\left({\rm{Tr}}({\bf{T}}^1_{z,k}{\bf{S}}^1_{z,k})-{\rm{Tr}}({\bf{T}}^1_{z,k}{{\hat{\bf{H}}}}_{z}{\bf{F}}{\bf{V}}_k{\bf{F}}^H{{\hat{\bf{H}}}}_{z}^H)\right)\nonumber\\
	+&\sum_{z=1}^{Z}\sum_{k=1}^{K}\left({\rm{Tr}}({\bf{T}}^2_{z,k}{\bf{S}}^2_{z,k})+{\rm{Tr}}\left({\bf{T}}^2_{z,k}{{\hat{\bf{H}}}}_{z}\left(\sum\nolimits_{i\neq k}^K{\bf{F}}{\bf{V}}_i{\bf{F}}^H\!+\!{\bf{F}}{\bf{\Lambda}}{\bf{F}}^H\right){{\hat{\bf{H}}}}_{z}^H\right)\right)+\sum_{k=0}^{K}{\rm{Tr}}({{\bf{\Omega}}_k\bf{V}}_k)+\psi'.
	\end{align}
	where
	${\bf{\Pi}}_2=\left\{{\bf{V}}_0,\{{\bf{V}}_k\},{\bf{\Lambda}},\{\varepsilon_k\},\{\lambda_k\},\{\theta_k\},\{\zeta_k^z\},\{\mu_k^z\},\{\tau_k\},\omega,\{\beta_k\},\{\gamma_k^z\},\psi_1,\psi_2,\{\psi_3^k\},\{\psi_4^k\},\{\psi_7^k\},\{\psi_8^k\},\{\psi_9^{l}\},\{{\bf{\Omega}}_k\},{\bf{T}}^1_{z,k}\right.$,\\$\left.{\bf{T}}^2_{z,k},\psi' \right\}$ and $\psi'$ denotes other Lagrangian terms that do not affect our analysis, ${\bf{S}}^1_{z,k}=\left[ \begin{array}{ccc}
	\kappa_z{\bf{I}}& {\bf{0}} \\
	{\bf{0}} & {\hat{\zeta}_z^k}\!-\!\kappa_z\sigma_z {\hat{\bf{h}}_z^{\rm{e}}}({\hat{\bf{h}}_z^{\rm{e}}})^H
	\end{array} 
	\right ]$, ${{\hat{\bf{H}}}}_{z}=\left[ \begin{array}{ccc}
	{\bf{I}} \\{\hat{\bf{h}}_z^{\rm{e}}}
	\end{array} \right ]$, ${\bf{S}}^2_{z,k}=\left[ \begin{array}{ccc}
	\upsilon_z{\bf{I}} & {\bf{0}} \\
	{\bf{0}}  &W_{\rm{mm}}N_0\!+\!{\chi_z^k}-\upsilon_z\sigma_z {\hat{\bf{h}}_z^{\rm{e}}}({\hat{\bf{h}}_z^{\rm{e}}})^H
	\end{array} \right ]$.
	\setcounter{equation}{\value{mytempeqncnt}}
\end{figure*}

Before dealing with~(\ref{C91}) and (\ref{C94}), we formulate the following  Lemma~\ref{lemma1} for the classic $\mathcal{S}$-Procedure~\cite{39Zhou}.
\begin{lemma}\label{lemma1}
	Define the following function
	\begin{eqnarray}
	F_i({\bf{x}})={\bf{x}}{\bf{A}}_i{\bf{x}}^H+2{\rm{Re}}\{{\bf{b}}_i{\bf{x}}^H\}+c_i, i\in\{1,2\},
	\end{eqnarray}
	where ${\bf{x}}\in{\mathbb{C}}^{1\times P}$, ${\bf{A}}_i\in{\mathbb{C}}^{P\times P}$,${\bf{b}}_i\in{\mathbb{C}}^{1\times P}$, $c_i\in\mathbb{R}$ and $P$ represents any integer. Then, the expression $F_1({\bf{x}})\leq 0\Rightarrow F_2({\bf{x}})\leq 0$ holds if and only if there exists a $\kappa$ satisfying
	\begin{eqnarray}\label{C121}
	\kappa\left[ \begin{array}{ccc}
		{\bf{A}}_1 & {\bf{b}}_1^H \\
		{\bf{b}}_1 & c_1
	\end{array} 
	\right ]-\left[ \begin{array}{ccc}
	{\bf{A}}_2 & {\bf{b}}_2^H \\
	{\bf{b}}_2 & c_2
	\end{array} 
	\right ]\succeq {\bf{0}}.
\end{eqnarray}  
\end{lemma} 

Additionally, (\ref{C91}) can be rewritten as
\begin{eqnarray}\label{C13}
\begin{aligned}
&\bigtriangleup{\bf{h}}_z^{\rm{e}}{\bf{F}}{\bf{V}}_k{\bf{F}}^H(\bigtriangleup{\bf{h}}_z^{\rm{e}})^H+2{\rm{Re}}\{{\hat{\bf{h}}}_z^{\rm{e}}{\bf{F}}{\bf{V}}_k{\bf{F}}^H(\bigtriangleup{\bf{h}}_z^{\rm{e}})^H\}\\
&\;\;\;\;\;\;\;\;\;\;\;\;\;+{\hat{\bf{h}}}_z^{\rm{e}}{\bf{F}}{\bf{V}}_k{\bf{F}}^H({\hat{\bf{h}}}_z^{\rm{e}})^H
-{\hat{\zeta}_z^k}\leq 0, k\!\in\! \mathcal{K},z\!\in\! \mathcal{Z}.
\end{aligned}
\end{eqnarray}

According to ${\bigtriangleup{\bf{h}}_z^{\rm{e}}}({\bigtriangleup{\bf{h}}_z^{\rm{e}}})^H\leq \sigma_z \left({\hat{\bf{h}}_z^{\rm{e}}}({\hat{\bf{h}}_z^{\rm{e}}})^H\right)$, (\ref{C13}) can be expressed as the following LMI with Lemma~\ref{lemma1}
	\begin{eqnarray}\label{C14}
\left[ \begin{array}{ccc}
\kappa_z{\bf{I}}\!-\!{\bf{F}}{\bf{V}}_k{\bf{F}}^H & -({\hat{\bf{h}}}_z^{\rm{e}}{\bf{F}}{\bf{V}}_k{\bf{F}}^H)^H \\
-{\hat{\bf{h}}}_z^{\rm{e}}{\bf{F}}{\bf{V}}_k{\bf{F}}^H & {\hat{\zeta}_z^k}\!-\!\kappa_z\sigma_z {\hat{\bf{h}}_z^{\rm{e}}}({\hat{\bf{h}}_z^{\rm{e}}})^H\!-\!{\hat{\bf{h}}}_z^{\rm{e}}{\bf{F}}{\bf{V}}_k{\bf{F}}^H({\hat{\bf{h}}}_z^{\rm{e}})^H
\end{array} 
\right ]\!\succeq\! {\bf{0}},
\end{eqnarray}  
where $\{\kappa_z\}$ is the slack variable.
Then, we rewrite $\Xi_k^z$  in (\ref{C15}) at the bottom of the next page.  Similar to (\ref{C13}) and (\ref{C14}), combining~Lemma~\ref{lemma1}, (\ref{C94}) can be expressed as the  LMI constraint (\ref{C16}) at the bottom of the next page.

For the OF (\ref{OptG0}), and the non-convex constraints~(\ref{OptA1}) and~(\ref{OptB1}) in problem~(\ref{OptG}), the method proposed in Section  III.~B can be used. To this end, (\ref{OptG}) is transformed into the following optimization problem
\setcounter{equation}{65}
\begin{subequations}\label{OptHA}
	\begin{align}
		\;\;\;\;\;\;\;&\underset{{\bf{\Omega}}}{\rm{max}}\;\sum_{k=1}^{K}W_{\rm{mm}}\hat{f}\left(\beta_k,\hat{\gamma}_k^{z},\hat{\gamma}_k^{z[n]}\right) \label{OptH0}\\
	&{\rm{s.t.}}\;\;{\rm(\ref{OptB4})-(\ref{OptB6}), (\ref{C12}), (\ref{C22}), (\ref{C61})},(\ref{C71}),(\ref{C81}),(\ref{C101}),\nonumber\\ &{\;\;\;\;\;\;\rm(\ref{C111}),({\ref{C14}}),(\ref{C16})}.
	\end{align}
\end{subequations}
where ${\bf{\Omega}}=\left\{{\bf{V}}_0,\{{\bf{V}}_k\},{\bf{\Lambda}},\{\varepsilon_k\},\{\lambda_k\},\{\theta_k\},\{\hat{\zeta}_k^z\},\{\hat{\mu}_k^z\},\{\hat{\chi}_k^z\},\{\tau_k\},\omega\right.$,
$\left.\{\beta_k\},\{\hat{\gamma}_k^z\},\{\kappa_z\},\{\upsilon_z\}\right\}$, $\hat{f}\left(\beta_k,\hat{\gamma}_k^{z},\hat{\gamma}_k^{z[n]}\right)=\log\left(1+\beta_k\right)-\log\left(1+\hat{\gamma}_k^{z[n]}\right)-\frac{\hat{\gamma}_k^z-\hat{\gamma}_k^{z[n]}}{1+\hat{\gamma}_k^{z[n]}}$.

The above optimization problem  will be convex if there is no rank-one constraint, and it can be solved by standard convex optimization techniques, such as interior-point method. Therefore,  by invoking rank-one relaxation, we can iteratively solve~(\ref{OptHA}) to obtain the solution of the original problem~(\ref{OptF}), following a similar procedure to that of  solving problems~(\ref{OptD}) and~(\ref{OptE}). 

Next, we study the rank relaxation problem. We first define the Lagrangian function  of the relaxed version of (\ref{OptHA}), which is expressed as (\ref{lagA}) at the bottom of the next page. 

 For ${\bf{V}}_0$, we have the same theorem to Theorem~\ref{theorem1} as well as the following Theorem~\ref{theorem3}.
\setcounter{equation}{67}
\begin{theorem}\label{theorem3}
	When $\psi_1^\ast\!>\!0, \psi_3^{i\ast}\!-\!\psi_8^{i\ast}\!\geq\! 0\; (i\!\neq\! k), {\bf{F}}^H{{\hat{\bf{H}}}}_{z}^H{\bf{T}}^{1\ast}_{z,k}{{\hat{\bf{H}}}}_{z}{\bf{F}}-\sum_{i\neq k}^{K}{\bf{F}}^H{{\hat{\bf{H}}}}_{z}^H{\bf{T}}^{2\ast}_{z,k}{{\hat{\bf{H}}}}_{z}{\bf{F}}\succeq {\bf{0}}$, we have 
	\begin{eqnarray}
	{\rm{rank}}({{\bf{V}}}_k^\ast)=1, k\in{\mathcal{K}},
	\end{eqnarray}
	where $\psi_1^\ast, \psi_3^{i\ast}, \psi_8^{i\ast}, {\bf{T}}^{1\ast}_{z,k},{\bf{T}}^{2\ast}_{z,k}$ denote the optimal Lagrange multipliers for the dual problem of~(\ref{OptHA}).
\end{theorem}
\begin{proof}
	Refer to Appendices~\ref{appendixA} and~\ref{appendixB}.
\end{proof}

In addition, the reconstruction of rank-one solution can refer to Appendices~\ref{appendixC} and~\ref{appendixD}.

Let us now analyze the complexity of solving~(\ref{OptHA}). Given an iteration accuracy $\epsilon$, the number of iterations is on the order $\sqrt{\Delta}\ln(1/\epsilon)$, where $\Delta=5KZ\!+\!2KZMN\!+\!KL\!+\!7K\!+\!N\!+\!1$ denotes the threshold parameters related to the constraints~\cite{43Com}. Additionally, (\ref{OptD}) includes $(KZ+3K+1)$ linear constraints, ($KZ+K$) two-dimensional LMI constraints, one  $N$-dimensional LMI constraint, $2KZ$  $(MN+1)$-dimensional  LMI constraints,  $K$  $L$-dimensional  LMI constraints and $K$ second order cone constraints. Therefore, the complexity of solving~(\ref{OptD}) is on the order of 
\begin{eqnarray}
{\mathcal{O}}\left({\overline{\xi}}\sqrt{\Delta}\ln(1/\epsilon)({\overline{\xi}}_1+{\overline{\xi}}_2\xi+{\overline{\xi}}_3+{\overline{\xi}}^2)\right),
\end{eqnarray} 
where ${\overline{\xi}}={\mathcal{O}}(N^2+KL^2)$ and $N^2+KL^2$ denotes the number of decision variables,  ${\overline{\xi}}_1=2KZ(MN+1)^3+N^3+KL^3+9KZ+11K+1$, ${\overline{\xi}}_2=2KZ(MN+1)^2+N^2+KL^25KZ+7K+1$, ${\overline{\xi}}_3=K(L+2)^2$. 

\begin{table} [t]
	\caption{Default Parameters Used In Simulation.} 
	\renewcommand{\arraystretch}{0.8}
	\label{Table I} 
	\centering
	\begin{tabular}{l|r} 
		\hline  
		\bfseries Parameters & \bfseries Value \\ [0.5ex] 
		\hline\hline
		Number of TAs at the CP  &$N=32$ \\
		\hline
		Number of TAs at each BS  &$M=4$ \\
		\hline
		Number of the BSs&$L=6$\\
		\hline
		Number of the users& $K=4$\\
		\hline
		Number of the Eves &$Z=2$\\
		\hline
		Number of the clusters & $C=4$\\
		\hline
		Distance between the CP and the BS cluster & 500  [m] \\
		\hline
		Radius of the BS cluster & 30 [m] \\
		\hline
		MmWave bandwidth & $W_{\rm{mm}}=50$ [MHz]  \\
		\hline
		Microwave bandwidth & $W_{\rm{mc}}=20$ [MHz]\\
		\hline  
		Standard deviation~\cite{44Channel} & $\sigma=4.6$ [dB]\\
		\hline
		Path loss at the mmWave frequency~\cite{44Channel} & $69.7+24\log_{10}(d)+X_{\sigma}$ [dB] \\ 
		\hline
		Path loss at the microwave frequency~\cite{40Zhang} & $38+30\log_{10}(d)$ [dB] \\ 
		\hline
		Noise variance~\cite{21Tao}& -174 [dBm/Hz] \\
		\hline 
	\end{tabular}
\end{table}

\section{Numerical Results}
In this section, numerical results are presented to evaluate the performance of the proposed algorithms.  {In addition, according to~\cite{44Channel}, the path-loss of the mmWave channel model can be expressed~as:
${\rm{PL}}(d)[{\rm{dB}}]={\rm{PL}}(d_0)+10{\overline{n}}\log_{10}(\frac{d}{d_0})+X_{\sigma}, {\rm{for}}\;d\geq d_0,$
where ${\rm{PL}}(d_0)$ is the close-in free-space path-loss in dB, $d_0$ is the close-in free-space reference distance, $\overline{n}$ is the best-fit minimum mean square error path-loss exponent over all measurement from a particular measurement campaign, and finally $X_{\sigma}$  is zero-mean Gaussian random variable with a standard deviation $\sigma$ in dB, representing the large-scale signal fluctuations resulting from shadowing owing to obstructions, such as buildings}. In our simulations, we set the RF carrier frequency to be 73 GHz, and the reference distance to $d_0=1$m. Correspondingly, the path-loss becomes ${\rm{PL}}(d_0)=69.7$ and ${\overline{n}}=2.4$. The default simulation parameters are listed in Table~\ref{Table I}.  Unless otherwise specified, these default values are used in simulation. The BSs, users and Eves are randomly located in the cell radius following a uniform distribution. The azimuth angle of departure at each BS is uniformly distributed over $[0, 2\pi]$. Additionally, we assume that the imperfect CSI of all Eves have the same error bound for simplicity.  {We plot the deployment figure of BSs, users and Eves as showed in~Fig.~\ref{figure1}.}
\begin{figure}
	\begin{center}
		\includegraphics[width=7cm,height=8cm]{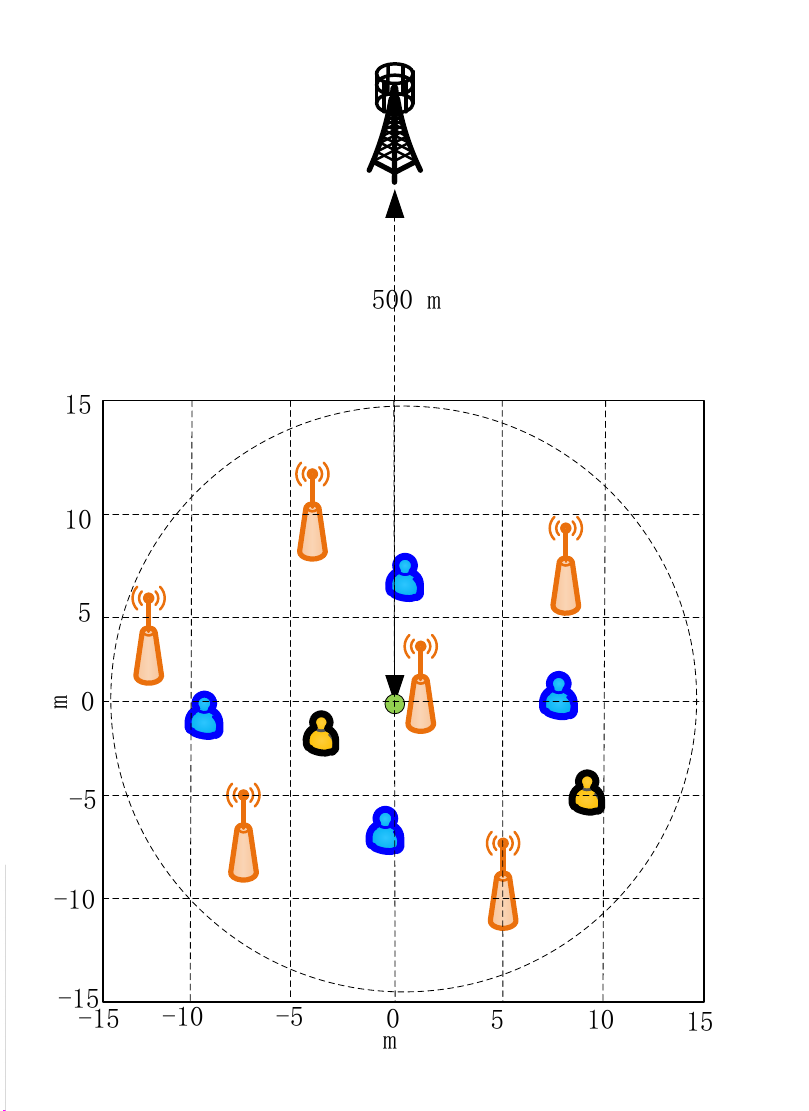}
		\caption{Simulation scenario with 6 BSs, 4 users and 2 Eves randomly distributed.}
	    \label{figure1}
	\end{center}
\end{figure} 

\begin{figure}[t]
	\begin{center}
		\includegraphics[width=9cm,height=7cm]{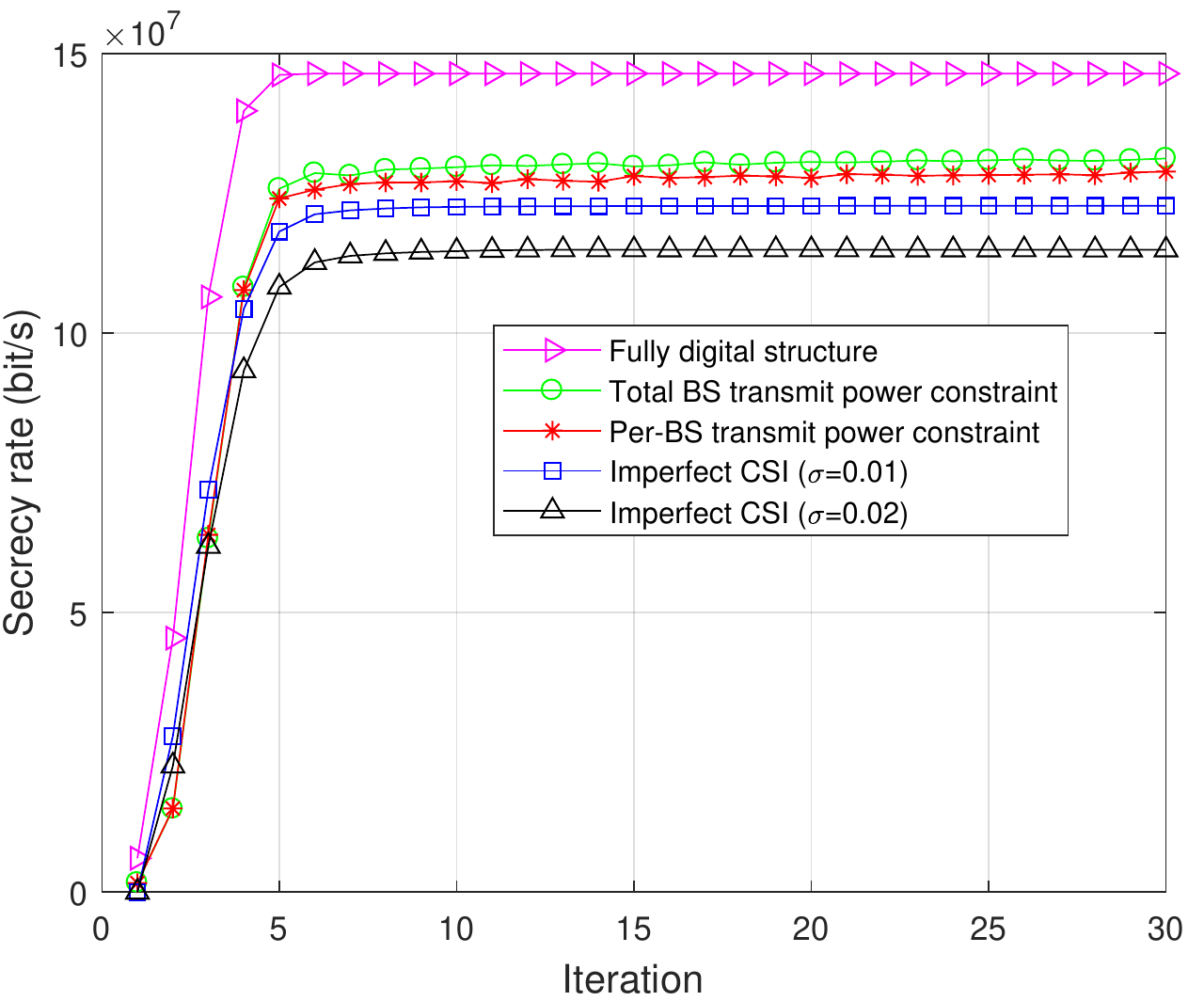}
		\caption{Secrecy  rate versus the number of iterations, $P_{\rm{max}}^{\rm{BS}}=15$ dBm and $P_{\rm{max}}^{\rm{AC}}=46$ dBm.}
		\label{Iteration}
	\end{center}
\end{figure}
\begin{figure}[t]
	\begin{center}
		\includegraphics[width=9cm,height=7cm]{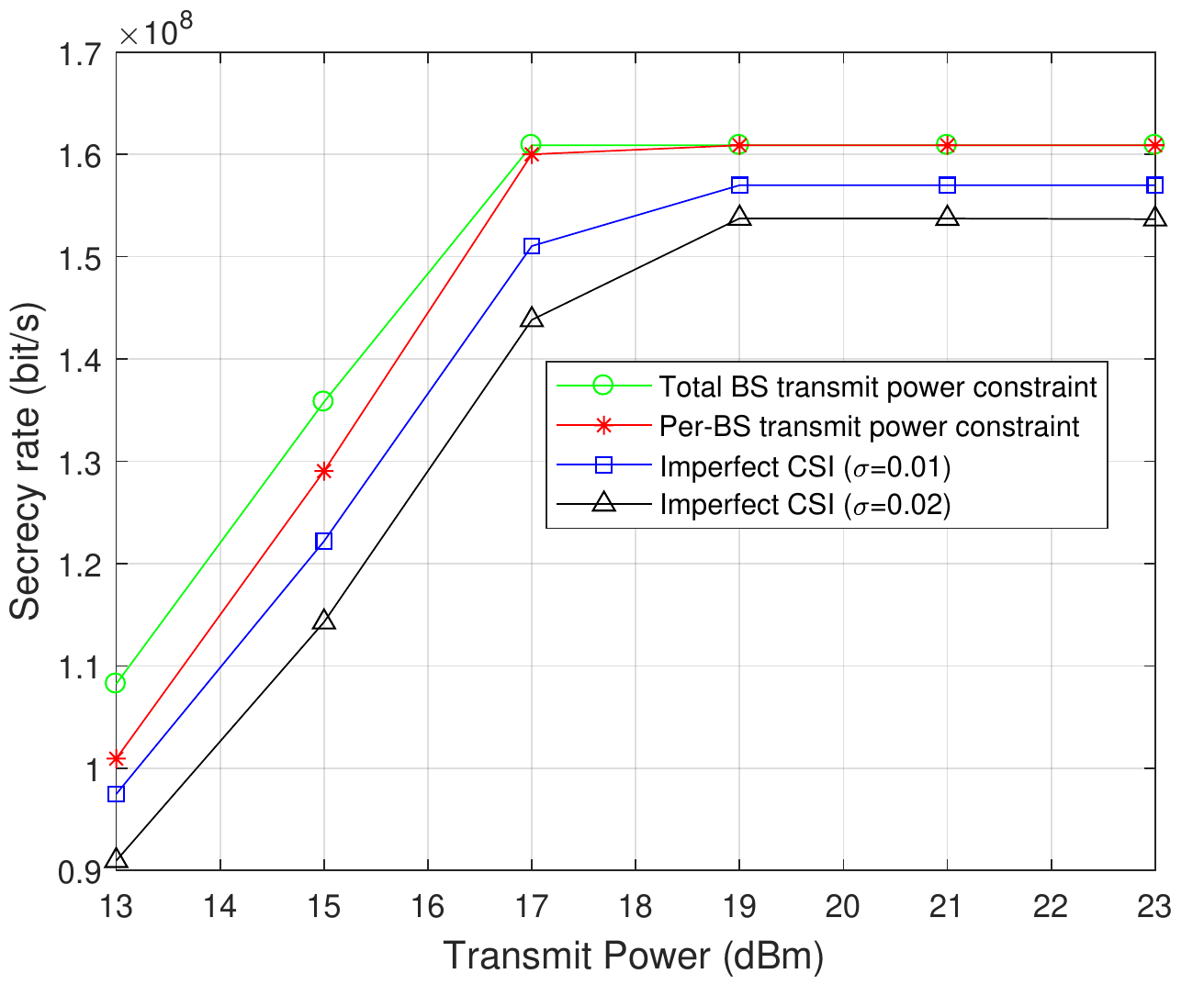}
		\caption{Secrecy rate  versus the total affordable transmit power for $L$ BSs, $P_{\rm{max}}^{\rm{CP}}=46$ dBm.}
		\label{BSPower}
	\end{center}
\end{figure}
\begin{figure}[t]
	\begin{center}
		\includegraphics[width=9cm,height=7cm]{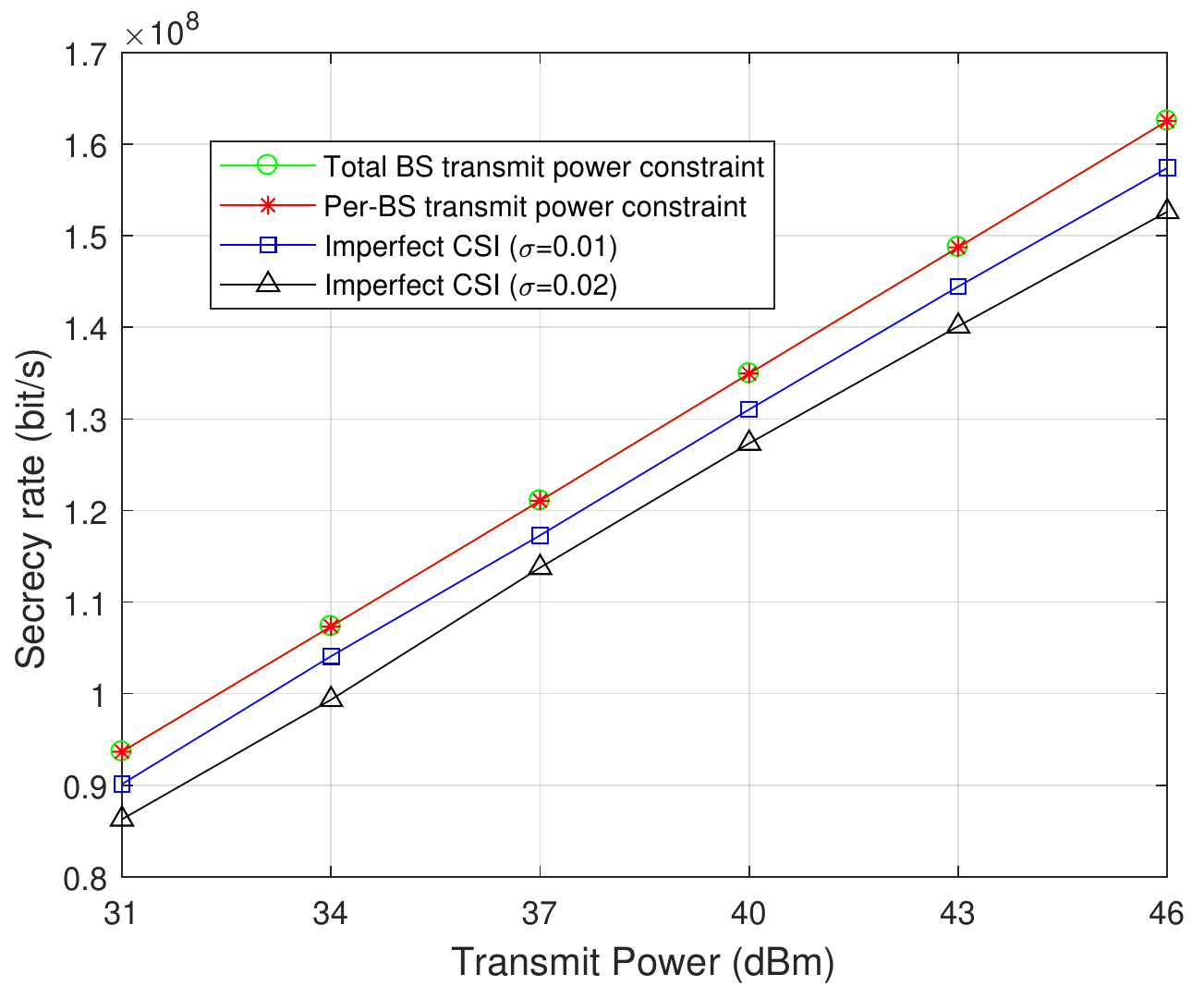}
		\caption{Secrecy rate versus the affordable transmit power of the CP, $P_{\rm{max}}^{\rm{BS}}=20$ dBm.}
		\label{CPPowerBS20}
	\end{center}
\end{figure}
\begin{figure}[t]
	\begin{center}
		\includegraphics[width=9cm,height=7cm]{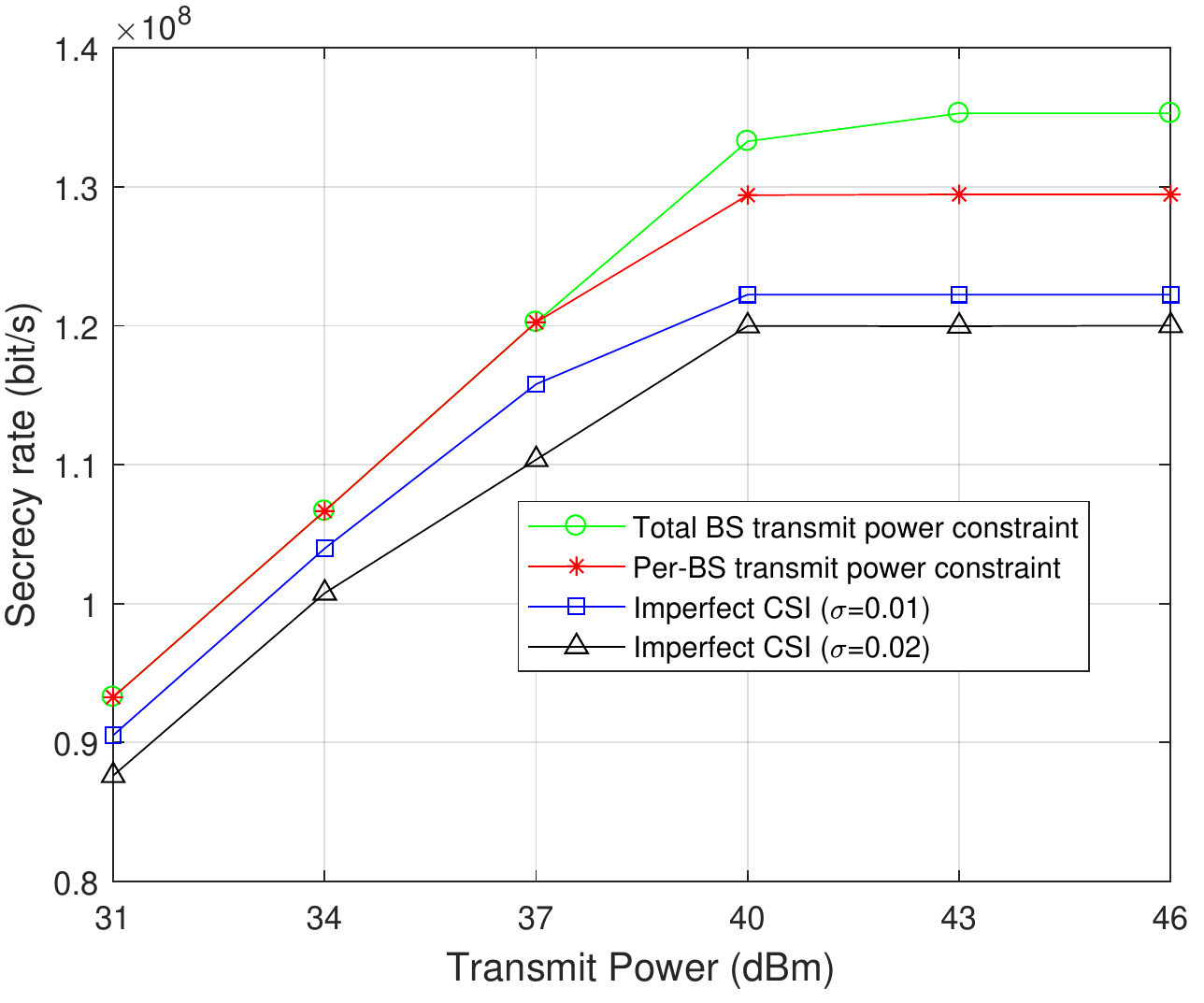}
		\caption{Secrecy rate versus the affordable transmit power of the CP, $P_{\rm{max}}^{\rm{BS}}=15$ dBm.}
		\label{CPPowerBS15}
	\end{center}
\end{figure}
\begin{figure}[t]
	\begin{center}
		\includegraphics[width=9cm,height=7cm]{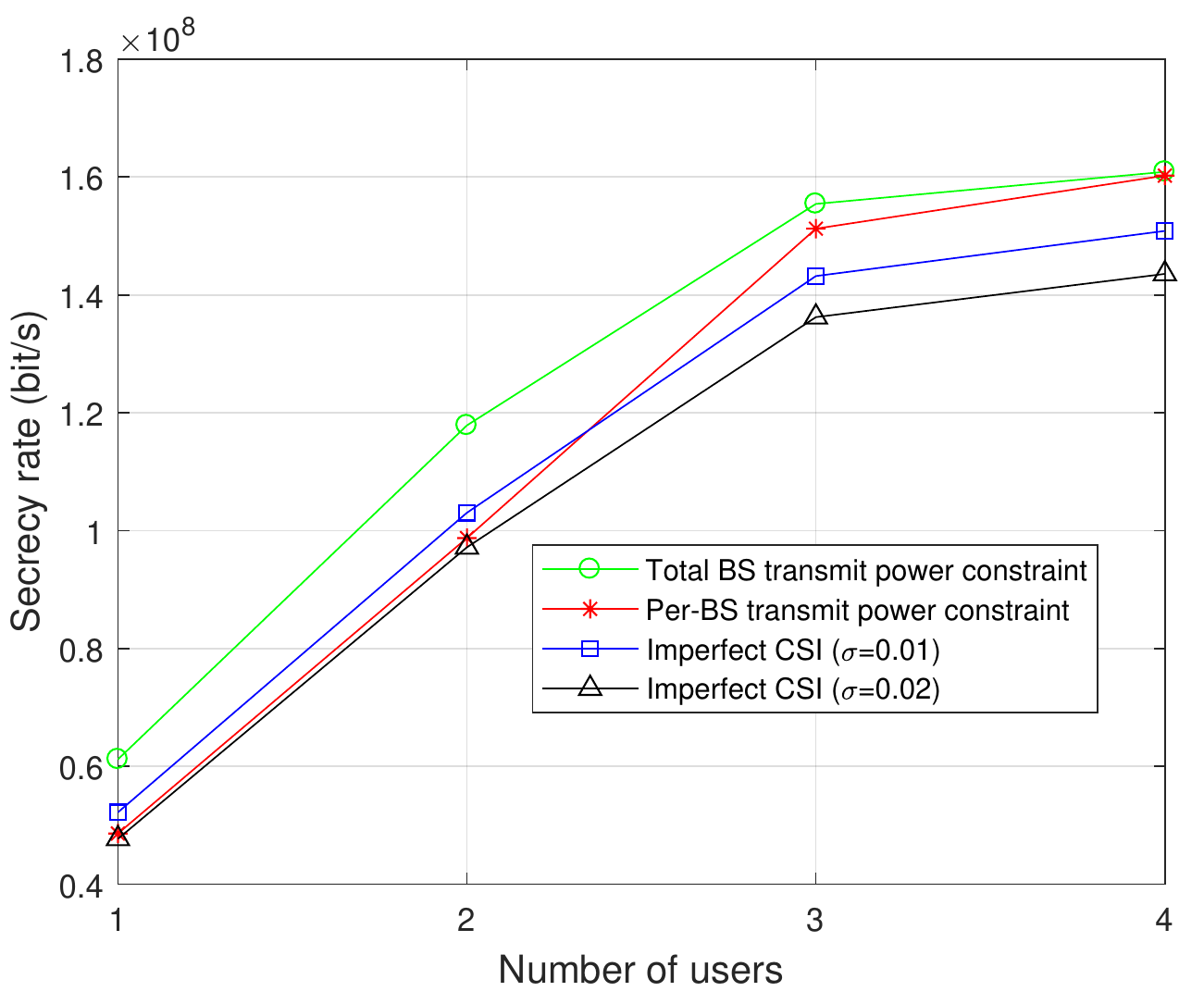}
		\caption{Secrecy rate versus the number of users, $P_{\rm{max}}^{\rm{BS}}=17$ dBm and $P_{\rm{max}}^{\rm{AC}}=46$ dBm.}
		\label{UserNumber}
	\end{center}
\end{figure}
\begin{figure}[t]
	\begin{center}
		\includegraphics[width=9cm,height=7cm]{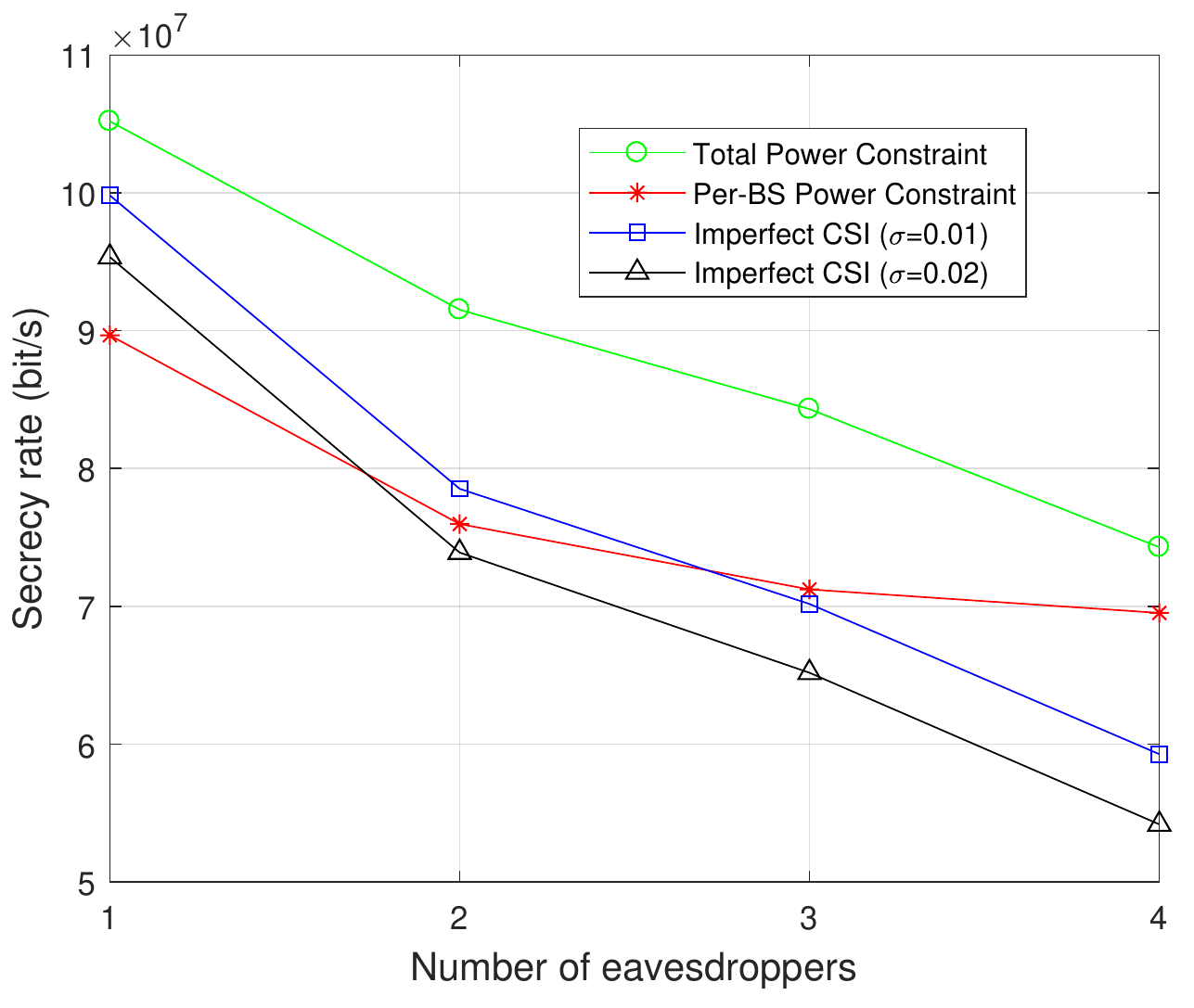}
		\caption{Secrecy rate versus the number of Eves,  $P_{\rm{max}}^{\rm{BS}}=15$, $P_{\rm{max}}^{\rm{AC}}=46$ dBm, $K=2$.}
		\label{EUserNumber}
	\end{center}
\end{figure}

Fig.~\ref{Iteration} demonstrates the convergence property of the different schemes, where we set $P_{\rm{max}}^{\rm{BS}}=15$ dBm and $P_{\rm{max}}^{\rm{AC}}=46$~dBm. For the per-BS transmit power constraint,  the total BS transmit power is equally allocated to each BS. Observe that the secrecy rate converges for all schemes. Additionally, the secrecy rate under the total BS transmit power constraint is slightly higher than that under per-BS transmit power constraint. This is because the affordable transmit power  of each BS  is fixed and limited under the per-BS transmit power constraint. However, the power can be flexibly allocated among BSs according to the CSI so that it improves the secrecy rate under the total BS transmit power constraint. As expected,  the secrecy rate decreases as the channel estimation error increases. Finally, we compare the performance under our adopted single RF chain structure to that under the fully digital structure where means that each antenna is connected to a single RF chain.  Although the secrecy rate is the highest for the fully digital structure, the associated hardware cost and energy consumption are high.

Fig.~\ref{BSPower} plots the secrecy rate versus the total BS affordable transmit power of the different schemes, and  we set $P_{\rm{max}}^{\rm{CP}}=46$ dBm. It can be observed that the secrecy rate first increases and then becomes saturated for all four cases. Although the secrecy rate can be increased as the total BS transmit power increases, the rate provided by the fronthaul link is limited due to the CP's limited transmit power. Therefore, even when the access link can potentially provide higher rate, the effective secrecy rate will be limited by the fronthaul link. Additionally, one can observe that the secrecy rate under the total BS transmit power constraint is higher than that under the per-BS transmit power constraint when the affordable transmit power is low, whereas they are almost the same when  the total BS affordable transmit power is high. {In fact, the affordable transmit power is indeed dissipated when it is lower}. However, as total BS affordable transmit power increases, it is unnecessary for each BS to dissipate all the affordable transmit power for increasing the rate, and this is because the fronthaul link does not provide  a sufficiently high rate due to the CP's limited transmit power. As a  result, the secrecy rates of the two schemes are the same.

In Fig.~\ref{CPPowerBS20}, we show the secrecy rate versus the CP's affordable transmit power, and we set $P_{\rm{max}}^{\rm{BS}}=20$ dBm. We observe that the secrecy rate increases as the CP transmit power increases for all schemes. Furthermore, the secrecy rate under the total BS and per-BS transmit power constraints is the same. This is because that when $P_{\rm{max}}^{\rm{BS}}=20$  and $P_{\rm{max}}^{\rm{CP}}\in[31:46]$ dBm, the  rate provided by the fronthaul link is still lower than the rate of the access link. To characterize the access link, Fig.~\ref{CPPowerBS15} shows the secrecy rate versus the CP's affordable transmit power at $P_{\rm{max}}^{\rm{BS}}=15$ dBm. Observe that the secrecy rate first increases and then saturates as the CP transmit power increases since  the BSs cannot provide a higher rate due to their limited transmit power. In this case, the effective secrecy rate cannot be improved even when the fronthaul link can  support a higher rate. 

Fig.~\ref{UserNumber} shows the secrecy rate versus the number of users under four schemes, where we set $P_{\rm{max}}^{\rm{BS}}=17$ dBm and $P_{\rm{max}}^{\rm{AC}}=46$ dBm. Indeed, it is expected that more users result in a higher sum secrecy rate. However, the secrecy rate is reduced when the channel estimation error is increased. Meanwhile, the gap of  secrecy rate between the total BS and per-BS  transmit power constraints becomes smaller as the number of users increases. In fact, when a single user supported by the BS cluster is close to  any  one of the BSs,  it will  be far away from the other BSs. In this case, the other BSs must overcome the large path loss to serve this user. However, as  the number of users increases, each BS may serve several users in its vicinity, hence the BSs tend to allocate more power to these ``near" users for improving their secrecy rate when operating under the  per-BS transmit power constraint.  

Finally, in Fig.~\ref{EUserNumber}, we plot the secrecy rate versus the number of Eves for different schemes, where  we set $P_{\rm{max}}^{\rm{BS}}=15$ dBm, $P_{\rm{max}}^{\rm{AC}}=46$ dBm and $K=2$. As expected,  the secrecy rate decreases as the number of Eves increases. Apparently, more Eves lead to higher wiretap rate, hence automatically decreasing the secrecy rate. Additionally, the secrecy rate gap between the total BS and per-BS  transmit power constraints becomes smaller, as the number of Eves increases. The reason is similar to that in Fig.~\ref{UserNumber}.

\section{Conclusions}
In this paper, we have investigated a secure BF design problem in the BS-cooperation-aided mmWave C-RAN relying on a microwave multicast fronthaul. We designed an advanced analog BF design scheme and we jointly optimized the multicast BF, digital BF and the AN covariance for maximizing the secrecy rate under the total BS and per-BS transmit power constraints. On this basis, we proposed convex approximation techniques and a CCCP-based iterative algorithm for solving them.  The perfect CSI assumption of the Eves links was subsequently replaced by imperfect CSI and the worst case SRM problem was considered. Finally,  an $S$-procedure-based iterative algorithm was developed. Our results have shown that the secrecy rate is the highest under the total BS transmit power constraint, and the computational complexity of solving the above problem is also the lowest, which is quantified by ${\mathcal{O}}\left(\xi\sqrt{\Delta}\ln(1/\epsilon)(\xi_1+\xi_2\xi+\xi_3+\xi^2)\right)$. Although the secrecy rate is relatively low under the imperfect CSI assumption for the Eves links compared to the other schemes and the computational complexity of solving the formulated problem is the highest given by ${\mathcal{O}}\left({\overline{\xi}}\sqrt{\Delta}\ln(1/\epsilon)({\overline{\xi}}_1+{\overline{\xi}}_2\xi+{\overline{\xi}}_3+{\overline{\xi}}^2)\right)$, the problem considered  is more practical and meaningful.

\appendices
\section{Proof of Theorem~\ref{theorem1}}\label{appendixA}
It is clear that the SDP-relaxed problem~(\ref{OptD}) is jointly convex with respect to the optimization variables, and thus it satisfies Slater's constraint qualification. Therefore, the Karush-Kuhn-Tucker (KKT) conditions are necessary and sufficient conditions for the optimal solutions of problem~(\ref{OptD}). Next, we focus our attention on the KKT conditions related to the optimal $\{{\bf{V}}_k^\ast\}$ and ${\bf{V}}_0^\ast$, namely:
 \begin{subequations}\label{Lag2}
	\begin{align}
     	\;\;\;\;\;\;\;\psi_2^\ast{\bf{I}}-\sum_{l=1}^L\psi_9^{l\ast}{\bf{G}}_l&={\bf{\Omega}}_0^\ast,\label{Lag21}\\
     \;\;\;\;\;{\bf{Y}}_k-\psi_4^{k\ast}{\overline{{\bf{H}}}}_k&={\bf{\Omega}}_k^\ast,k\in{\mathcal{K}},\label{Lag22}\\
     \;\;\;\;\;\psi_9^{l\ast}\left({\rm{Tr}}({\bf{G}}_l{\bf{V}}_0^{\ast})-W_{\rm{mc}}N_0(e^{\omega^\ast/\eta}-1)\right)&=0,l\in{\mathcal{L}},\label{Lag23}\\
     \;\;\;\;\;{\bf{\Omega}}_k^\ast{\bf{V}}_k^\ast&={\bf{0}},k\in\{0,\mathcal{K}\},\label{Lag24}\\
     \;\;\;\;\;{\bf{\Omega}}_k^\ast&\succeq{\bf{0}},k\in\{0,\mathcal{K}\},\label{Lag25}
	\end{align}
\end{subequations} 
where ${\bf{Y}}_k=\psi_1^\ast{\bf{I}}+\sum_{i\neq k}^{K}(\psi_3^{i\ast}-\psi_8^{i\ast}){\overline{{\bf{H}}}}_i+\sum_{z=1}^Z\left(\psi_5^{z,k\ast}-\sum_{i\neq k}\psi_6^{z,i\ast}\right){\rm{Tr}}({\overline{{\bf{H}}^e}}_z)+\psi_7^{k\ast}{\overline{{\bf{H}}}}_k$, and $\{\psi_1^\ast,\psi_2^\ast,\{\psi_3^{k\ast}\},\{\psi_4^{k\ast}\},\{\psi_5^{z,k\ast}\},\{\psi_6^{z,k\ast}\}$,
$\{\psi_7^{k\ast}\},\{\psi_8^{k\ast}\},\{\psi_9^{l\ast}\},\{{\bf{\Omega}}_k^\ast\}$ are the optimal Lagrange multipliers for the dual problem of~(\ref{OptD}).  

First, we prove ${\rm{rank}}({{\bf{V}}}_0^\ast)\leq L$. Multiplying both sides of
(\ref{Lag21}) by ${{\bf{V}}}_0^\ast$ and combing (\ref{Lag24}), we have
\begin{eqnarray}
 \psi_2^\ast{{\bf{V}}}_0^\ast=\sum_{l=1}^L\psi_9^{l\ast}{\bf{G}}_l{{\bf{V}}}_0^\ast.
\end{eqnarray}
When $\psi_2^\ast>0$, the following equation must  hold 
\begin{eqnarray}
	{\rm{rank}}({{\bf{V}}}_0^\ast)=	{\rm{rank}}(\psi_2^\ast{{\bf{V}}}_0^\ast)={\rm{rank}}\left(\sum_{l=1}^L\psi_9^{l\ast}{\bf{G}}_l{{\bf{V}}}_0^\ast\right),
\end{eqnarray}
as
${\rm{rank}}\left(\sum_{l=1}^L\psi_9^{l\ast}{\bf{G}}_l{{\bf{V}}}_0^\ast\right)\leq {\rm{min}}\left\{{\rm{rank}}\left(\sum_{l=1}^L\psi_9^{l\ast}{\bf{G}}_l\right),{\rm{rank}}\left({{\bf{V}}}_0^\ast\right)\right\},$
and thus we have 
\begin{eqnarray}\label{rankAA}
{\rm{rank}}({{\bf{V}}}_0^\ast)=	{\rm{rank}}\left(\sum_{l=1}^L\psi_9^{l\ast}{\bf{G}}_l\right)={\rm{rank}}\left(\sum_{l=1}^L\psi_9^{l\ast}{\bf{g}}_l^H{\bf{g}}_l\right)\leq L.
\end{eqnarray}

Next, we prove that there always exists a ${\bf{V}}_0^\ast$ so that ${\rm{rank}}({{\bf{V}}}_0^\ast)\leq \sqrt{L}$. When the optimal solutions of problem~(\ref{OptD}) are obtained, there exists at least an $l$ satisfying $\omega^\ast= \eta\log\left(1+\frac{{\rm{Tr}}({\bf{G}}_l{\bf{V}}_0^\ast)}{W_{\rm{mc}}N_0}\right)$. We first consider a special case, namely $\omega^\ast= \eta\log\left(1+\frac{{\rm{Tr}}({\bf{G}}_l{\bf{V}}_0^\ast)}{W_{\rm{mc}}N_0}\right)$ for any $l$, and we have
\begin{eqnarray}\label{rankl}
	{\rm{Tr}}({\bf{G}}_l{\bf{V}}_0^{\ast})=W_{\rm{mc}}N_0(e^{\omega^\ast/\eta}-1),l\in{\mathcal{L}}.
\end{eqnarray}

Let $\mathcal{R}={\rm{rank}}({{\bf{V}}}_0^\ast)$. Then ${{\bf{V}}}_0^\ast$ can be decomposed as ${{\bf{V}}}_0^\ast={{\bf{X}}}{{\bf{X}}}^H$, ${{\bf{X}}}\in \mathbb{C}^{N\times \mathcal{R}}$, and we have
\begin{eqnarray}\label{rank2}
{\rm{Tr}}({\bf{G}}_l{\bf{V}}_0^{\ast})={\rm{Tr}}({{\bf{X}}}^H{\bf{G}}_l{\bf{X}})=W_{\rm{mc}}N_0(e^{\omega^\ast/\eta}-1),l\in{\mathcal{L}}.
\end{eqnarray}
We define $\Gamma$ as a  $\mathcal{R}$-by-$\mathcal{R}$ Hermitian matrix, and formulate the following linear equation:
\begin{eqnarray}\label{rank3}
{\rm{Tr}}({{\bf{X}}}^H{\bf{G}}_l{\bf{X}}\Gamma)=0,l\in{\mathcal{L}}.
\end{eqnarray}

For $\Gamma$, there are $\mathcal{R}(\mathcal{R}+1)/2$ unknown real parts and $\mathcal{R}(\mathcal{R}-1)/2$ unknown imaginary parts, and thus the total number of real-value variables is $\mathcal{R}^2$. Therefore, there are $\mathcal{R}^2$ unknowns and $L$ equations in (\ref{rank3}).

If $\mathcal{R}^2>L$, there must exist a nonzero solution for linear (\ref{rank3}). We denote $\varrho_l (l\in\mathcal{L})$ as eigenvalues of $\Gamma$, and define 
\begin{eqnarray}
	\varrho^\ast=\max \{|\varrho_l|:l\in\mathcal{L}\},
\end{eqnarray}
and it is easy to obtain the following matrix
\begin{eqnarray}
	{\bf{I}}-\frac{1}{\varrho^\ast}\Gamma\succeq {\bf{0}}.
\end{eqnarray}
Let us define ${{\hat{\bf{V}}}}_0^\ast={{\bf{X}}}({\bf{I}}-\frac{1}{\varrho^\ast}\Gamma){{\bf{X}}}^H$,  which has the following properties~\cite{44Rank2}:
\begin{itemize}
\item Rank of ${{\hat{\bf{V}}}}_0^\ast$ at least reduces one: ${\rm{rank}({{\hat{\bf{V}}}}_0^\ast)}\leq \mathcal{R}-1$.
\item Primal feasibility:
${\rm{Tr}}({\bf{G}}_l{\hat{\bf{V}}}_0^{\ast})=W_{\rm{mc}}N_0(e^{\omega^\ast/\eta}-1),l\in{\mathcal{L}}.$
\item Power constraint:
${\rm{Tr}}({\hat{\bf{V}}}_0^\ast)\leq P_{\rm{max}}^{\rm{AC}}$.
\end{itemize}

According to the above properties, it the clear that the feasible solution ${{\hat{\bf{V}}}}_0^\ast$ is also the optimal solution for problem~(\ref{OptD}). If we have ${\rm{rank}^2({{\hat{\bf{V}}}}_0^\ast)}>L$, repeat the above procedure, else, stop. Then we have ${\rm{rank}^2({{\hat{\bf{V}}}}_0^\ast)}\leq L$, namely ${\rm{rank}({{\hat{\bf{V}}}}_0^\ast)}\leq\sqrt{L}$.

If there are $L\!-\!1$ equations satisfying $\omega^\ast= \eta\log\left(1+\frac{{\rm{Tr}}({\bf{G}}_l{\bf{V}}_0^\ast)}{W_{\rm{mc}}N_0}\right)$, we can use the above method to obtain ${\rm{rank}({{\hat{\bf{V}}}}_0^\ast)}\leq\sqrt{L-1}$. Therefore, there always exists a ${\bf{V}}_0^\ast$ such that 
${\rm{rank}}({{\bf{V}}}_0^\ast)\leq \sqrt{L}$. 

Next, we prove that if there exists any $l$ so that  $\|{\bf{g}}_{l}\|<\|{\bf{g}}_{l'}\|(l'\in\{1,\dots,l-1,l+1,\dots,L\})$, we have
${\rm{rank}}({{\bf{V}}}_0^\ast)=1$. According to the above assumption, it is easy to obtain the following conclusions
 \begin{subequations}
 	\begin{align}
 	{\rm{Tr}}({\bf{G}}_l{\bf{V}}_0^{\ast})&=W_{\rm{mc}}N_0(e^{\omega^\ast/\eta}-1),\\
 	{\rm{Tr}}({\bf{G}}_{l'}{\bf{V}}_0^{\ast})&> W_{\rm{mc}}N_0(e^{\omega^\ast/\eta}\!-\!1),l'\in\{1,\!\dots\!,l\!-\!1,l\!+\!1,\!\dots\!,L\}.
 	\end{align}
 \end{subequations}

Combining (\ref{Lag23}), we can obtain $\psi_9^{l\ast}>0$ and $\psi_9^{l'\ast}=0$. From (\ref{rankAA}), the following equation holds 
\begin{eqnarray}\label{rankAB}
{\rm{rank}}({{\bf{V}}}_0^\ast)=	{\rm{rank}}\left(\sum_{l=1}^L\psi_9^{l\ast}{\bf{G}}_l\right)={\rm{rank}}\left(\sum_{l=1}^L\psi_9^{l'\ast}{\bf{g}}_l^H{\bf{g}}_l\right)=1,
\end{eqnarray}
which concludes the proof.
\section{Proof of Theorem~\ref{theorem2}}\label{appendixB}
According to~(\ref{Lag22}), when $\psi_1^\ast>0, \psi_3^{i\ast}-\psi_8^{i\ast}\geq 0\;(i\neq k), \psi_5^{z,k\ast}-\sum_{i\neq k}^{K}\psi_6^{z,i\ast}\geq 0$, ${\bf{Y}}_k$ is a positivedefinite matrix which has full rank, namely we have ${\rm{rank}}({\bf{Y}}_k)=L$. Then  we have 
\begin{eqnarray}\label{rank4}
\begin{aligned}
	{\rm{rank}}({\bf{\Omega}}_k^\ast)&={\rm{rank}}({\bf{Y}}_k-\psi_4^{k\ast}{\overline{{\bf{H}}}}_k)\geq {\rm{rank}}({\bf{Y}}_k)-{\rm{rank}}(\psi_4^{k\ast}{\overline{{\bf{h}}}}_k^H{\overline{{\bf{h}}}}_k)\\
	&\geq  L-1.
\end{aligned}
\end{eqnarray} 
While according to~(\ref{rank4}), the  rank of ${\bf{\Omega}}_k^\ast$ is either $L$ or $L-1$. If ${\rm{rank}}({\bf{\Omega}}_k^\ast)=L$, ${\bf{V}}^\ast_k={\bf{0}}$ due to~(\ref{Lag24}), which means that the users cannot receive any signal. Thus, we have ${\rm{rank}}({\bf{\Omega}}_k^\ast)=L-1$, and the null space of ${\bf{\Omega}}_k^\ast$ has a single dimension. Meanwhile, (\ref{Lag24}) indicates that ${\bf{\Omega}}_k^\ast$ must lie in the null-space of ${\bf{\Omega}}_k^\ast$. Therefore, it must hold ${\rm{rank}}({\bf{\Omega}}_k^\ast)=1$, which concludes the proof.
\section{The Rank-One Reconstruction For ${\bf{V}}_k^\ast$ }\label{appendixC}
Let us introduce the notation of  ${\rm{rank}}({\bf{Y}}_k)=r$. In Appendix~\ref{appendixB}, we assume $\psi_1^\ast>0, \psi_3^{i\ast}-\psi_8^{i\ast}\geq 0\;(i\neq k), \psi_5^{z,k\ast}-\sum_{i\neq k}^{K}\psi_6^{z,i\ast}\geq 0$, and $r=L$. When the above assumptions do not hold, we have $r<L$, which implies that ${\bf{Y}}_k$ is not a full-rank matrix. Then, we define ${\bf{\Upsilon}}_k\in{\mathbb{C}}^{L\times L-r}$ as the orthogonal basis of the null space of ${\bf{Y}}_k$, where ${\bf{\Upsilon}}_k^H{\bf{\Upsilon}}_k={\bf{I}}$, and  the following must hold:
\begin{eqnarray}
	{\bf{Y}}_k{\bf{\Upsilon}}_k=0,\;\;\;{\rm{rank}}({\bf{\Upsilon}}_k)=L-r.
\end{eqnarray}

Let ${\bf{d}}_{k,i}\in{\mathbb{C}^{L\times 1}}$ denote the $i$th column of ${\bf{\Upsilon}}_k$, where $1\leq i\leq L-r$. Then, we have
\begin{eqnarray}
	{\bf{d}}_{k,i}^H{\bf{\Omega}}_k^\ast{\bf{d}}_{k,i}={\bf{d}}_{k,i}^H({\bf{Y}}_k-\psi_4^{k\ast}{\overline{{\bf{h}}}}_k^H{\overline{{\bf{h}}}}_k){\bf{d}}_{k,i}=-\psi_4^{k\ast}|{\overline{{\bf{h}}}}_k{\bf{d}}_{k,i}|^2.
\end{eqnarray}

According to~(\ref{C22}) and the KKT condition~(\ref{Lag2}), it is clear that $\psi_4^{k\ast}>0$ and ${\bf{\Omega}}_k^\ast\succeq {\bf{0}}$, and thus
 ${\overline{{\bf{h}}}}_k^H{\overline{{\bf{h}}}}_k{\bf{\Upsilon}}_k={\bf{0}}$.  As a result, we have
 ${\bf{\Omega}}_k^\ast{\bf{\Upsilon}}_k=({\bf{Y}}_k-\psi_4^{k\ast}{\overline{{\bf{h}}}}_k^H{\overline{{\bf{h}}}}_k){\bf{\Upsilon}}_k={\bf{0}}.$
Combining~${\rm{rank}}({\bf{\Upsilon}}_k)=L-r$, we have ${\rm{rank}}({\bf{\Omega}}_k)\leq r$.
In addition, it can be shown from~(\ref{Lag22}) that ${\rm{rank}}({\bf{\Omega}}_k^\ast)\geq {\rm{rank}}({\bf{Y}}_k)-{\rm{rank}}(\psi_4^{k\ast}{\overline{{\bf{h}}}}_k^H{\overline{{\bf{h}}}}_k)=r-1$. Therefore, we have 
$L-r	\leq {\rm{rank}}({\bf{V}}_k^\ast)\leq  L-r+1.$
However, when ${\rm{rank}}({\bf{V}}_k^\ast)={\rm{rank}}({\bf{\Upsilon}}_k^\ast)=L-r$, we have ${\bf{V}}_k^\ast={\bf{\Upsilon}}_k^\ast$ via~(\ref{Lag23}) and~(\ref{Lag24}), and ${\bf{V}}_k^\ast$ can be expressed as ${\bf{V}}_k^\ast=\sum\nolimits_{i=1}^{L-r}a_{k,i}{\bf{d}}_{k,i}{\bf{d}}_{k,i}^H (a_{k,i}>0)$. However, this means that the user cannot receive the signal due to ${\overline{{\bf{h}}}}_k^H{\overline{{\bf{h}}}}_k{\bf{\Upsilon}}_k={\bf{0}}$, and thus we have ${\rm{rank}}({\bf{V}}_k^\ast)=  L-r+1$. Therefore, there exists only a single subspace that lies in the null space of ${\bf{\Omega}}_k^\ast$, which is denoted as ${\bf{d}}_{k,0}$. Then, we can write ${\bf{V}}_k^\ast=\sum\nolimits_{i=1}^{L-r}a_{k,i}{\bf{d}}_{k,i}{\bf{d}}_{k,i}^H+a_{k,0}{\bf{d}}_{k,0}{\bf{d}}_{k,0}^H$. Finally, we reconstruct the new feasible $\{{\hat{\bf{V}}}_0^\ast,\{{\hat{\bf{V}}}_k^\ast\}, {\hat{\bf{\Lambda}}}^\ast\}$~as
\begin{subequations}\label{rank5}
\begin{align}
\;\;\;\;\;\;\;\;\;\;\;\;\;\;{\hat{\bf{V}}}_0^\ast=&{{\bf{V}}}_0^\ast,\\
\;\;\;\;\;\;\;\;\;\;\;\;\;\;{\hat{\bf{V}}}_k^\ast=&{\bf{V}}_k^\ast-\sum\nolimits_{i=1}^{L-r}a_{k,i}{\bf{d}}_{k,i}{\bf{d}}_{k,i}^H=a_{k,0}{\bf{d}}_{k,0}{\bf{d}}_{k,0}^H,\\
\;\;\;\;\;\;\;\;\;\;\;\;\;\;{\hat{\bf{\Lambda}}}^\ast =&{{\bf{\Lambda}}^\ast}+\sum\nolimits_{i=1}^{L-r}a_{k,i}{\bf{d}}_{k,i}{\bf{d}}_{k,i}^H.
\end{align}
\end{subequations} 

Substituting  $\{{\hat{\bf{V}}}_0^\ast,\{{\hat{\bf{V}}}_k^\ast\}, {\hat{\bf{\Lambda}}}^\ast\}$ into~(\ref{OptD}), we can arrive of the optimal value as the optimal solution, while satisfying all the constraints.
\section{The Rank-One Reconstruction For ${\bf{V}}_0^\ast$ }\label{appendixD}
For $\{{\bf{V}}_k^\ast\}(k\in\{{\mathcal{K}}\})$, the optimal BF $\{{\bf{v}}_k^\ast\}$ can be directly obtained using the eigenvalue decomposition (EVD) method. Next, we adopt the randomization technique of~\cite{41Rank1} for  obtaining the rank-one ${\bf{V}}_0^\ast$. First, by applying the EVD technique, we decompose $\{{\bf{V}}_0^\ast\}$~as
\begin{eqnarray}
	{\bf{V}}_0^\ast={{\bf{X}}}_0^\ast{{\bf{D}}}_0^\ast({{\bf{X}}}_0^\ast)^H.
\end{eqnarray}
We form the $i$th candidate BF vector ${{\bf{v}}}_0^{i\ast}={{\bf{X}}}_0^\ast{{\bf{D}}}_0^{\ast{1}/{2}}{{\bf{s}}}_i$, where ${{\bf{s}}}_i\sim {\mathcal{CN}}({\bf{0}},{\bf{I}})$. As a result, we have $\mathbb{E}\{{{\bf{v}}}_0^{i\ast}({{\bf{v}}}_0^{i\ast})^H\}={\bf{V}}_0^\ast$. For the $i$th candidate BF, we reformulate the following optimization problem
 \begin{subequations}\label{OptJ}
 	\begin{align}
 		&\underset{\left\{a_0^i,\{a_k^i\},{\bf{\Lambda}}\right\}}{\rm{max}}\;\sum_{k=1}^{K}\left[R_{k,i}^{{\rm{AC}}}-\underset{z\in{\mathcal{Z}}}{\max}\left\{R_{k,z,i}^{\rm{EV}}\right\}\right]^+ \label{OptJ0}\\
 	{\rm{s.t.}}&\;\; \sum_{k=1}^{K}R_{k,i}^{\rm{AC}}\leq \underset{l\in{\mathcal{L}}}{\min}\;\left\{R_{l,i}^{\rm{FH}}\right\},\label{OptJ1}\\
 	&\;\;\sum_{k=1}^{K}a_k^{i\ast}||{\bf{F}}{\bf{v}}_k^i||^2+{\rm{Tr}}({\bf{F}}^H{\bf{F}}{\bf{\Lambda}})\leq P_{\rm{max}}^{\rm{BS}},\label{OptJ2}\\
 	&\;\;a_0^i||{\bf{v}}_0^{i\ast}||^2\leq P_{\rm{max}}^{\rm{AC}},\label{OptJ3}
 	\end{align}
 \end{subequations} 
where 
\begin{eqnarray}
\begin{aligned}
R_{l,i}^{\rm{FH}}&=W_{\rm{mc}}\log\left(1+\frac{a_0^i|{\bf{g}}_l{\bf{v}}_0^{i\ast}|^2}{W_{\rm{mc}}N_0}\right),\\
R_{k,i}^{\rm{AC}}\!&=\!W_{\rm{mm}}\log\left(1\!+\!\frac{a_k^i|{\bf{h}}_k{\bf{F}}{\bf{v}}_k^{i\ast}|^2}{\sum\nolimits_{m\neq k}^Ka_m^i|{\bf{h}}_k{\bf{F}}{\bf{v}}_m^{i\ast}|^2\!+\!{\bf{h}}_k{\bf{F}}{\bf{\Lambda}}({\bf{h}}_k{\bf{F}})^H\!+\!W_{\rm{mm}}N_0}\right),\\
R_{k,z,i}^{\rm{EV}}\!&=\!W_{\rm{mm}}\log\left(1\!+\!\frac{a_k^i|{\bf{h}}_z^{\rm{e}}{\bf{F}}{\bf{v}}_k^{i\ast}|^2}{\sum\nolimits_{m\neq k}^Ka_m^i|{\bf{h}}_z^{\rm{e}}{\bf{F}}{\bf{v}}_m^{i\ast}|^2\!+\!{\bf{h}}_z^{\rm{e}}{\bf{F}}{\bf{\Lambda}}({\bf{h}}_z^{\rm{e}}{\bf{F}})^H\!+\!W_{\rm{mm}}N_0}\right).\nonumber
\end{aligned}
\end{eqnarray}

The above problem can be solved by using the method proposed in~Section III.~B, hence we omit the process due to the limited space. Finally, we select the optimal $\{a_k^{i\ast}\}$ from all the candidate BF vectors that owns the maximum secure rate, and the optimal BF vector can be expressed as~${\bf{v}}_k^{\rm{*}}=\sqrt{a_k^{i\ast}}{\bf{v}}_k^{{i\ast}}(k\in\{0,{\mathcal{K}}\})$.

\end{document}